\documentclass[journal,twoside,web]{ieeecolor}
\usepackage{generic}
\usepackage{graphicx}          

\usepackage{textcomp}

\usepackage{amsmath}
\usepackage{amssymb}
\usepackage{amsfonts}
\usepackage{subcaption}

\usepackage{algorithm}

\usepackage{algpseudocode}

\usepackage{mathtools}

\let\labelindent\relax 
\usepackage{enumitem}

\usepackage[dvipsnames]{xcolor}
\usepackage{tikz}
\usetikzlibrary{3d,calc}
\usepackage{pgfplots}
\usepackage{bm,comment}
\usepgfplotslibrary{fillbetween}
\usetikzlibrary{arrows,hobby,tikzmark}
\usetikzlibrary{decorations.markings,decorations.pathreplacing}
\usetikzlibrary{math}

\newcommand{\tr}{^\intercal}
\newcommand{\norm}[1]{\left\| #1 \right\|}

\DeclarePairedDelimiterX{\Set}[2]\{\}{%
  \, #1 \;\delimsize\vert\; #2 \,
}

\usepackage{amsthm}
\newtheoremstyle{customthm}
  {1.5mm}
  {1.5mm}
  {\itshape}
  {}
  {\bfseries}
  {.}
  {1.5mm}
  {}
  
\theoremstyle{customthm}
\newtheorem{Theorem}{Theorem}
\newtheorem{Proposition}{Proposition}

\newtheorem{Lemma}{Lemma}
\newtheorem{Definition}{Definition}
\newtheorem{Remark}{Remark}

\allowdisplaybreaks

\begin{document}

\title{Once upon a time step: A closed-loop approach to robust MPC design }

\author{Anilkumar Parsi, Marcell Bartos, Amber Srivastava, Sebastien Gros, Roy S. Smith, \IEEEmembership{Fellow, IEEE}
\thanks{This work was supported by the Swiss National Science Foundation under Grant 200021\_178890 and NCCR Automation Grant 180545.
Corresponding author Anilkumar Parsi.}
\thanks{Anilkumar Parsi, Amber Srivastava, Roy Smith and Marcell Bartos are with the Automatic Control Laboratory, Swiss Federal Institute of Technology (ETH Z\"urich), 8092 Zurich, Switzerland (e-mail:
aparsi/asrivastava/rsmith@control.ee.ethz.ch, mbartos@student.ethz.ch).}
\thanks{Sebastien Gros is with the Department of Engineering Cybernetics, Norwegian University of Science and Technology (NTNU), Trondheim, Norway (email:sebastien.gros@ntnu.no)}
}

\maketitle

\begin{abstract}                          
A novel perspective on the design of robust model predictive control (MPC) methods is presented, 
whereby closed-loop constraint satisfaction is ensured using recursive feasibility of the MPC optimization. Necessary and sufficient conditions are derived for recursive feasibility, based on the effects of model perturbations and disturbances occurring at {\em one time step}. Using these conditions and Farkas' lemma, sufficient conditions suitable for design are formulated. The proposed method is called a closed-loop design, as only the existence of feasible inputs at the next time step is enforced by design. This is in contrast to most existing formulations, which compute control policies that are feasible under the worst-case realizations of all model perturbations and exogenous disturbances in the MPC prediction horizon. The proposed method has an online computational complexity similar to nominal MPC methods while preserving guarantees of constraint satisfaction, recursive feasibility and stability. Numerical simulations demonstrate the efficacy of our proposed approach.
\end{abstract}


\section{Introduction}


\IEEEPARstart{M}{odel} predictive control (MPC) is one of the popular control strategies in use today. Owing to the modeling and computational flexibility, MPC has been successfully implemented in a diverse range of applications --- such as oil refineries \cite{qin2003survey}, autonomous vehicles \cite{williams2018information}, robotics \cite{Pieters2016model}, and supply chain management \cite{lejarza2021economic}. A distinctly highlighting feature of MPC is its ability to efficiently handle constraints on the system states and control inputs, while also guaranteeing the stability of the system \cite{rawlings2017model}. In particular, MPC poses decision making as an optimization problem, where one optimizes an objective function subject to the constraints resulting from the system dynamics, and admissible system states and control inputs. The optimization problem is solved at each time step to compute an open-loop input sequence over a finite number of time steps (called the prediction horizon), of which only the first input is applied in closed loop. Hence MPC produces an optimization-based closed-loop policy from the current system state to an input to be applied to the system. A popular approach in practice is to use an approximate linear dynamical model of the system, as it results in convex optimization problems suitable for real-time implementation \cite{darby2012mpc}. Such an approximation can be compensated
for, by introducing structured uncertainty and disturbances into the model.

Robust MPC algorithms explicitly consider the effects of model uncertainties and disturbances in the controller design so that closed-loop stability and constraint satisfaction can be guaranteed \cite{kouvaritakis2015model}. In addition to these properties, an important consideration for MPC is whether the optimization problem remains feasible in closed-loop operation, which is known as recursive feasibility \cite{kerrigan2001robust_PHD,lofberg2012oops}. The design of robust MPC thus has three desired properties: (i) constraint satisfaction, (ii) recursive feasibility, and (iii) closed-loop stability. In this paper, we highlight that generally (ii) already implies (i), and use this observation to propose a novel method for designing robust MPC schemes. 

The design for (i) and (ii) is coupled and follows an open-loop approach in most existing MPC works \cite{fleming2014robust,hu2019output,parsi2022scalable}. This involves two steps. First, constraint satisfaction is explicitly enforced on the open-loop trajectories computed by the MPC optimizer, considering the worst-case model perturbations and disturbances along the prediction horizon of MPC. Second, the state of the system at the end of the prediction horizon is driven into a precomputed terminal set, which is robustly invariant under a known terminal controller \cite{kouvaritakis2015model}. This strategy ensures that the open-loop trajectories do not violate constraints for an infinite time, thus achieving both recursive feasibility and constraint satisfaction.

Although such an open-loop design approach works, a large degree of conservatism can originate from ensuring feasibility over multiple time steps in the prediction horizon. 
This is because recursive feasibility of MPC is a one-step property. That is, recursive feasibility only requires that under the action of the first control input computed by MPC at the current time step, the MPC optimization is also feasible at the next time step. However, the open-loop design for recursive feasibility operates via ensuring feasibility throughout the prediction horizon, instead of the next time step alone. The conservatism of open-loop approaches is arguably stemming from this discrepancy, as depicted in Figure \ref{fig:CLR1_1}. 

In Figure \ref{fig:CLR1_1}, it can be seen that tube MPC approaches (such as \cite{mayne2005robust,schwenkel2022}) construct a state tube which accounts for all possible disturbances $w_k, w_{k+1}, w_{k+2}$ and model perturbations $\Delta_k,\Delta_{k+1},\Delta_{k+2}$ in the prediction horizon of three time steps, before the open-loop state trajectories reach the terminal set. 

In addition to the conservatism, many open-loop approaches have a large computational complexity of online optimization. This is because in an open-loop approach, MPC either has to recursively outer-approximate the reachable sets under worst-case perturbations (as in tube MPC methods), or use scenario trees which grow exponentially along the horizon \cite{maiworm2015scenario}. 

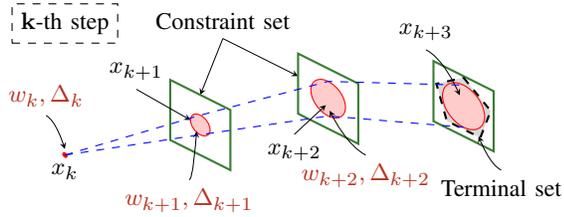
\begin{figure}
     \centering	
     \begin{tikzpicture}[font=\small]
 	\tikzset{axisArrow/.style={black,loosely dotted, {Latex[length=0.7mm,width=0.7mm]}-{Latex[length=0.7mm,width=0.7mm]}, thin}}
	\begin{scope}[shift={(-0.5,2)}]
	\node[draw,dashed] at (0,0) {$\mathbf{k}$-th step};
	\end{scope}
 	\begin{scope}[yslant=-0.5,shift={(0,0)}]
 		\coordinate (x0) at (-0.5,0);
 		\filldraw [red]  (x0) circle[radius=0.8pt] node[anchor=north,black] {$x_k$};
 		\coordinate (w0) at (-0.75,0.4);
 		\node [anchor=south,BrickRed] at (w0) {$w_k,\Delta_k$};
 		\draw[black,ultra thin,-stealth,anchor=north]  (w0)parabola(-0.5,0.08);
 	\end{scope}
 	\begin{scope}[yslant=-0.5,shift={(1.3,1.3)}]
 	{
 		\coordinate (x1_1) at (0,0.13);
 		\coordinate (x1_2) at (0,-0.13);
 		\coordinate (CS) at (0.3,1.3);
 		\coordinate (CS1) at (0,0.4);
 		\coordinate (CS2) at (1.3,1.3);
 		\coordinate (x1) at (-0.85,0.05);
 		\coordinate (w1) at (-0.15,-0.8);
 		\node [anchor=north,BrickRed] at (w1) {$w_{k+1},\Delta_{k+1}$};
 		\node [anchor=south] at (x1) {$x_{k+1}$};
 		\node [anchor=south] at (CS) {Constraint set};
 		\draw[black,ultra thin,-stealth,anchor=south]  (CS)--(CS1);
 		\draw[black,ultra thin,-stealth,anchor=south]  (CS)--(CS2);
 		\draw[OliveGreen,thick,fill=black!0] (-0.4,-0.4) rectangle +(0.8,0.8);
 		\draw[red,fill=red!20] (0,0) circle [radius=0.13];
 		\draw[black,ultra thin,-stealth] (x1)--(-0.13,0);
 		\draw[black,ultra thin,-stealth] (w1)parabola(-0.05,-0.11);
 	}
 	\end{scope}
 	\begin{scope}[yslant=-0.5,shift={(3.0,2.5)}]
 	{
 		\coordinate (x2_1) at (0,0.23);
 		\coordinate (x2_2) at (0,-0.23);
 		\coordinate (x2) at (-0.45,-0.7);
 		\coordinate (w2) at (0.5,-0.5);
 		\node [anchor=north] at (x2) {$x_{k+2}$};
 		\node [anchor=north,BrickRed] at (w2) {$w_{k+2},\Delta_{k+2}$};
 		\draw[OliveGreen,thick,fill=black!0] (-0.4,-0.4) rectangle +(0.8,0.8);
 		\draw[red,fill=red!20] (0,0) circle [radius=0.23];
 		\draw[black,ultra thin,-stealth] (x2)--(-0.03,-0.1);
 		\draw[black,ultra thin,-stealth] (w2)parabola(0.02,-0.15);
 	}
 	\end{scope}
 	\begin{scope}[yslant=-0.5,shift={(4.8,3.3)}]
 	{
 		\coordinate (x3_1) at (0,0.28);
 		\coordinate (x3_2) at (0,-0.28);
 		\coordinate (x3) at (-0.5,0.5);
 		\coordinate (TC) at (0.5,-0.6);
 		\node [anchor=north] at (TC) {Terminal set};
 		\node [anchor=south] at (x3) {$x_{k+3}$};
	 	\draw[OliveGreen,thick,fill=black!0] (-0.4,-0.4) rectangle +(0.8,0.8);
	 	\draw[black,thick,dashed,fill=black!15] (0,0.38)--(0.361,0.117)--(0.223,-0.307)--(-0.223,-0.307)--(-0.361,0.117)--cycle;
 		\draw[red,fill=red!20] (0,0) circle [radius=0.28];
 		\draw[black,ultra thin,-stealth] (x3)parabola(-0.05,0.12);
 		\draw[black,ultra thin,-stealth] (TC)parabola(0.2,-0.3);
 	}
 	\end{scope}
	\draw[blue, dashed] (x0)--(x1_1)--(x2_1)--(x3_1);
	\draw[blue, dashed] (x0)--(x1_2)--(x2_2)--(x3_2);
	\begin{scope}[yslant=-0.5,shift={(2.5,0.2)}]
	\end{scope}

	\end{tikzpicture}
     \caption{Classical tube MPC methods, where the predicted open-loop trajectories remain feasible under \emph{all} model perturbations $\Delta_i$ and disturbances $w_i$ occurring in the prediction horizon.}\label{fig:CLR1_1}
\end{figure}

This work is motivated by the insight that under mild conditions, recursive feasibility implies constraint satisfaction. That is, if an MPC problem remains feasible in closed loop, then the system constraints are also satisfied. Therefore, one need not explicitly ensure constraint satisfaction along the prediction horizon under the open-loop MPC policy. Although this observation is not particularly difficult to verify, it is not exploited in the controller design in most existing methods. 

 
To the best of our knowledge, this insight is not formally presented in previous literature, but a few early robust MPC methods implicitly use it for systems affected by additive disturbances \cite{chisci99Robustifying,kerrigan2001Robust,kerrigan2001robust_PHD}. 
All these works study the recursive feasibility of the MPC optimization and do not explicitly prove constraint satisfaction, as recursive feasibility already implies constraint satisfaction. The methods in \cite{chisci99Robustifying} 
 and Chapter 6.5 of \cite{kerrigan2001robust_PHD} propose ways to robustify nominal MPC controllers, if a robust control invariant set is known a priori. 
However, these sets are very difficult to compute \cite{rakovic2007optimized,rungger2017}. Whereas \cite{kerrigan2001Robust} studies the problem of ensuring recursive feasibility using set-invariance, it only provides tools for analysis of a given robust MPC algorithm. The recent work \cite{abdelsalam2021synthesis}  extends the ideas from \cite{kerrigan2001Robust} for systems with multiplicative perturbations, but also provides conditions suitable for analysis, not design. In contrast to these early methods, most later works explicitly design for constraint satisfaction following the open-loop approach described earlier  \cite{fleming2014robust}-\cite{schwenkel2022}. 

In this work, we present a robust MPC design which explicitly ensures recursive feasibility in closed-loop operation, and hence refer to it as closed-loop robust MPC. The proposed method in this work uses constraint tightening to design a recursively feasible MPC. In constraint tightening methods, the system constraints along the prediction horizon are tightened to account for model uncertainty, such that if a nominal trajectory satisfies the modified constraints, the closed-loop trajectories remain feasible under model perturbations and disturbances \cite{chisci2001systems,parsi2022computationally}. Because the online optimization computes only the nominal trajectories, the computational complexity is similar to that of nominal MPC schemes.

The closed-loop robust MPC proposed in this work is depicted in Figure \ref{fig:CLR1_2}. In this approach, the robust MPC design accounts {\em only} for the disturbance $w_k$ and model perturbation $\Delta_k$ occurring at the current ($k$-th) time step. In particular, the MPC constraints are tightened offline such that, for any feasible nominal trajectory at a time step $k$, disturbances $w_k$ and perturbations $\Delta_k$ result in a feasible nominal trajectory at time step $k+1$, guaranteeing recursive feasibility. Note that, although constraint tightening techniques already exist in literature (\cite{chisci2001systems,parsi2022computationally}), the tightenings therein are computed by accounting for all perturbations and disturbances in the MPC prediction horizon, instead of those occurring at the first time step as done in this work.

\begin{figure}
     \centering	
     \begin{tikzpicture}[font=\small]
 	\tikzset{axisArrow/.style={black,loosely dotted, {Latex[length=0.7mm,width=0.7mm]}-{Latex[length=0.7mm,width=0.7mm]}, thin}}
	\begin{scope}[shift={(-0.5,-1.5)}]
	    \node[draw,dashed] at (0,0) {$\mathbf{k}$-th step};
	\end{scope}
 	\begin{scope}[yslant=-0.5,shift={(0,-3.5)}]
 		\coordinate (x0) at (-0.5,0);
 		\draw[OliveGreen,thick,fill=black!0] (-0.9,-0.4) rectangle +(0.8,0.8);
 		\draw[BrickRed,dashed,fill=BrickRed!10] (-0.8,-0.3) rectangle +(0.6,0.6);
 		\coordinate (NT) at (0.3,1.3);
 		\coordinate (NT1) at (0.15,0.5);
 		\node[anchor=south] at (NT) {Nominal trajectory};
 		\draw[black,ultra thin,-stealth] (NT) parabola (NT1);
 		\filldraw [red]  (x0) circle[radius=0.8pt] node[anchor=north,black] {$x_k$};
 	\end{scope}
 	\begin{scope}[yslant=-0.5,shift={(3.0,-1.0)}]
 	{
 	    \coordinate (x2) at (0,0);
 		\draw[OliveGreen,thick,fill=black!0] (-0.4,-0.4) rectangle +(0.8,0.8);
 		\draw[ForestGreen,dashed,fill=ForestGreen!10] (-0.22,-0.22) rectangle +(0.44,0.44);
 		\filldraw [red]  (x2) circle[radius=0.8pt];
 	}
 	\end{scope}
 	\begin{scope}[yslant=-0.5,shift={(1.3,-2.3)}]
 	{
 	    \coordinate (x1) at (0,0);
 	    \coordinate (TiC1) at (0.3,-0.2);
 	    \coordinate (TiC2) at (1.8,1.11);
 	    \coordinate (TC) at (1.5,-0.2);
 		\node [anchor=north] at (TC) {Tightened constraints};
 		\draw[OliveGreen,thick,fill=black!0] (-0.4,-0.4) rectangle +(0.8,0.8);
 		\draw[Blue,dashed,fill=Blue!10] (-0.25,-0.25) rectangle +(0.5,0.5);
 		\draw[Black,ultra thin,-stealth] (TC) to (TiC1);
 		\draw[Black,ultra thin,-stealth] (TC) to (TiC2);
 		\filldraw [red]  (x1) circle[radius=0.8pt];
 	}
 	\end{scope}
 	\begin{scope}[yslant=-0.5,shift={(4.8,-0.2)}]
 	{
 	    \coordinate (x3) at (0,0);
 	    \coordinate (TC) at (0.5,-0.6);
 	    \node [anchor=north] at (TC) {Terminal set};
	 	\draw[OliveGreen,thick,fill=black!0] (-0.4,-0.4) rectangle +(0.8,0.8);
	 	\draw[black,thick,dashed,fill=black!15] (0,0.38)--(0.361,0.117)--(0.223,-0.307)--(-0.223,-0.307)--(-0.361,0.117)--cycle;
	 	\draw[black,ultra thin,-stealth] (TC)parabola(0.2,-0.3);
	 	\filldraw [red]  (x3) circle[radius=0.8pt];
 	}
 	\end{scope}
 	\begin{scope}[decoration={markings,mark=at position 0.5 with {\arrow{stealth}}}]
 	\draw[blue,dashed,postaction={decorate}] (x0)--(x1);
 	\draw[blue,dashed,postaction={decorate}] (x1)--(x2);
 	\draw[blue,dashed,postaction={decorate}] (x2)--(x3);
 	\end{scope}
 	
 	
    \begin{scope}[shift={(-0.5,-4.8)}]
	    \node[draw,dashed] at (0,0) {$\mathbf{k+1}$-th step};
	\end{scope}
 	\begin{scope}[yslant=-0.5,shift={(0,-6.5)}]
 		\coordinate (x0) at (-0.5,0);
 		\filldraw [red]  (x0) circle[radius=0.8pt] node[anchor=north,black] {$x_k$};
 		\coordinate (w0) at (-0.75,0.4);
 		\node [anchor=south,BrickRed] at (w0) {$w_k,\Delta_k$};
 		\draw[black,ultra thin,-stealth,anchor=north]  (w0)parabola(-0.5,0.08);
 	\end{scope}
 	\begin{scope}[yslant=-0.5,shift={(1.3,-5.3)}]
 	{
 		\coordinate (x1) at (0,0.1);
 		\coordinate (x1_loc) at (-0.75,-0.2);
 		\coordinate (x1_1) at (0,0);
 		\coordinate (x1_3) at (0,0.17);
 		\coordinate (x1_4) at (0,-0.17);
 		\node[anchor=south] at (x1_loc) {$x_{k+1}$};
 		\draw[OliveGreen,thick,fill=black!0] (-0.4,-0.4) rectangle +(0.8,0.8);
 		\draw[BrickRed,dashed,fill=BrickRed!10] (-0.3,-0.3) rectangle +(0.6,0.6);
 		\draw[red,fill=red!20] (0,0) circle [radius=0.17];
 		\filldraw [red]  (x1) circle[radius=0.8pt];
 		\draw[black,ultra thin,-stealth] (x1_loc)--(-0.02,0.1);
 		\filldraw [red]  (x1_1) circle[radius=0.8pt];
 	}
 	\end{scope}
 	\begin{scope}[yslant=-0.5,shift={(3.0,-4.0)}]
 	{
 	    \coordinate (x2) at (0,0);
 	    \coordinate (x2_3) at (0,0.15);
 	    \coordinate (x2_4) at (0,-0.15);
 		\draw[OliveGreen,thick,fill=black!0] (-0.4,-0.4) rectangle +(0.8,0.8);
 		\draw[Blue,dashed,fill=Blue!10] (-0.25,-0.25) rectangle +(0.5,0.5);
 		\draw[red,fill=red!20] (0,0) circle [radius=0.15];
 		\filldraw [red]  (x2) circle[radius=0.8pt];
 	}
 	\end{scope}
 	\begin{scope}[yslant=-0.5,shift={(4.8,-3.2)}]
 	{
 	    \coordinate (x3) at (0,0);
 	    \coordinate (x3_3) at (0,0.13);
 	    \coordinate (x3_4) at (0,-0.13);
	 	\draw[OliveGreen,thick,fill=black!0] (-0.4,-0.4) rectangle +(0.8,0.8);
	 	\draw[ForestGreen,dashed,fill=ForestGreen!10] (-0.22,-0.22) rectangle +(0.44,0.44);
	 	\draw[red,fill=red!20] (0,0) circle [radius=0.13];
	 	\coordinate (PNTx2) at (-1.8,-0.85);
 	    \coordinate (PNTx1) at (-0.02,-0.06);
 	    \coordinate (PNT) at (-1.5,-1.5);
 	    \node[anchor=north] at (PNT) {\begin{tabular}{c}Perturbed nominal\\ trajectory for $k+1$-th step\end{tabular}};
 	    \draw[black,ultra thin,-stealth] (PNT)--(PNTx1);
	 	\draw[black,ultra thin,-stealth] (PNT)--(PNTx2);
	 	\filldraw [red]  (x3) circle[radius=0.8pt];
 	}
 	\end{scope}
 	\begin{scope}[yslant=-0.5,shift={(6.5,-2.2)}]
 	{
 	    \coordinate (x4) at (0,0);
 	    \coordinate (x4_3) at (0,0.12);
 	    \coordinate (x4_4) at (0,-0.12);
	 	\draw[OliveGreen,thick,fill=black!0] (-0.4,-0.4) rectangle +(0.8,0.8);
	 	\draw[black,thick,dashed,fill=black!15] (0,0.38)--(0.361,0.117)--(0.223,-0.307)--(-0.223,-0.307)--(-0.361,0.117)--cycle;
	 	\draw[red,fill=red!20] (0,0) circle [radius=0.12];
	 	\filldraw [red]  (x4) circle[radius=0.8pt];
 	}
 	\end{scope}
 	\begin{scope}[decoration={markings,mark=at position 0.5 with {\arrow{stealth}}}]
     	\draw[black,postaction={decorate}] (x0)--(x1);
     	\draw[blue,dashed,postaction={decorate}] (x1_1)--(x2);
     	\draw[blue,dashed,postaction={decorate}] (x2)--(x3);
     	\draw[blue,dashed,postaction={decorate}] (x3)--(x4);
     	\draw[blue] (x1_3)--(x2_3)--(x3_3)--(x4_3);
     	\draw[blue] (x1_4)--(x2_4)--(x3_4)--(x4_4);
 	\end{scope}
 	\begin{scope}[yslant=-0.5,shift={(2.5,-6.4)}]
	\end{scope}
 	\end{tikzpicture}
     \caption{Closed-loop robust MPC (our proposed approach). Here, MPC only computes a nominal trajectory online such that tightened constraints are satisfied. The tightenings are designed offline such that, under perturbation $\Delta_k$ and disturbance $w_k$ occurring at the \emph{current} ($k$-th) time step, the MPC problem will be feasible at the next ($k+1$-th) time step.}\label{fig:CLR1_2}
\end{figure}
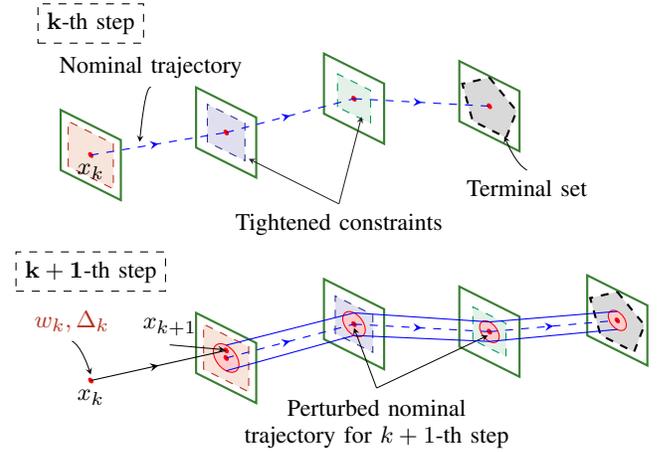

In this work, a constraint tightening technique is developed for systems whose dynamics can be described by a nominal linear model combined with a time-varying linear fractional perturbation \cite{doyle1991} and bounded exogenous disturbances. The proposed algorithm has two phases: offline and online. In the offline phase, a non-convex optimization is solved to compute constraint tightenings. In the online phase, a nominal MPC problem with tightened constraints is solved at each time step to compute the control inputs. The proposed framework is designed to ensure recursive feasibility, constraint satisfaction and input-to-state stability. Numerical examples demonstrate that the proposed design results in an order of magnitude reduction in online computation time and improved performance compared to state-of-the-art methods.

Our work has three novel contributions. First, we formally present the conditions under which recursive feasibility of MPC optimization implies closed-loop constraint satisfaction. Second, necessary and sufficient conditions for recursive feasibility are presented for systems affected by linear fractional perturbations and disturbances. Third, recursive feasibility is formulated as a design constraint on the robust MPC controller, by analyzing only the effects of perturbations and disturbances occurring at one time step.

\subsection{Notation}
The sets of real numbers is denoted by $ \mathbb{R} $ and the sequence of integers from $ n_1 $ to $ n_2 $  by $ \mathbb{N}_{n_1}^{n_2} $. 
For a vector $ b $,  $\norm{b}_2$ represents the $2-$norm. The $ i ^{th}$ row of a matrix $ A $ is denoted by $ [A]_{i} $. The notation $ a_{l|k} $ denotes the value of $ a $ at time step $ k+l $ computed at the time step $ k $. The identity matrix of size $n\times n$ is denoted by $I_n$.  A block of zeros of size $n\times m$ inside a larger matrix will be denoted by $0_{n,m}$ when the size cannot be inferred from the other blocks, and $0$ otherwise. The minimum eigenvalue of a matrix $A$ is denoted as $\rho_{\min}(A)$. The Kronecker product of two matrices $A$ and $B$ is denoted by $A \otimes B$, and $\text{co}\{\cdot \}$ represents the convex hull operator. A continuous function $ \alpha: \mathbb{R}_{\ge 0} \rightarrow \mathbb{R}_{\ge 0}  $ is a $ \mathcal{K} $ function if $ \alpha(0) = 0 $, $ \alpha(s) >0 $ for all $ s>0 $ and it is strictly increasing; it is a $ \mathcal{K}_{\infty} $ function if it is a $ \mathcal{K} $ function and $\alpha(s)\rightarrow \infty$ as $s\rightarrow \infty$. A continuous function $ \beta:\mathbb{R}_{\ge 0} \times \mathbb{N}_{0}^{\infty} \rightarrow \mathbb{R}_{\ge 0} $ is a $ \mathcal{K} \mathcal{L} $ function if $ \beta(s,t) $ is a $ \mathcal{K}$ function in $ s $ for every $ t\ge 0 $, it is strictly decreasing in $ t $ for every $ s>0 $ and $ \beta(s,t) \rightarrow 0$ when $ t \rightarrow \infty $.

\section{Problem formulation}
We consider uncertain linear, time-invariant systems of the form:
\begin{subequations}\label{eq:Dynamics}
\begin{align}
    x_{k+1} &= A x_{k} + B u_{k} + B_p p_k + B_w w_k, \label{eq:Dynamics1}\\
    q_k &= D_x x_k + D_u u_k + D_w w_k, \label{eq:Dynamics2}\\ 
    p_k &= \Delta_k q_k,\label{eq:Dynamics3} 
\end{align}
\end{subequations}
where $x_k \in \mathbb{R}^{n_x} $ represents the state of the system, $u_k \in \mathbb{R}^{n_u} $ represents the control input, and $w_k \in \mathbb{R}^{n_w}$ represents an exogenous disturbance acting on the system's state. In addition, the model uncertainty is captured using a linear fractional transformation (see \cite{doyle1991}), described by $p_k, q_k \in \mathbb{R}^{n_p}$ and the matrix $\Delta_k \in \mathbb{R}^{n_p \times n_p}$. The vectors $ p_k, q_k $ and the matrix $\Delta_k$ cannot be measured, but $\Delta_k$ is known to lie inside the set
\begin{align}\label{eq:DeltaBound}
    \mathcal{P}:= \text{co}\{ \Delta^1, \Delta^2,\ldots ,\Delta^{n_{\Delta}} \},
\end{align}
where $n_{\Delta}$ represents the number of vertices defining the convex hull $\mathcal{P}$. 

The exogenous disturbance $w_k$ is assumed to lie within a compact set containing the origin, defined as
\begin{align}\label{eq:wBound}
    \mathcal{W}:= \{ w \in \mathbb{R}^{n_w}|  H_w w \le h_w \},
\end{align}
where $ H_w \in \mathbb{R}^{n_{H_w}\times n_w}$ and $ h_w \in \mathbb{R}^{n_{H_w}}$. 

Moreover, the states and inputs of the system are required to lie in a compact polytopic constraint set containing the origin, defined as
\begin{align}\label{eq:Constraints}
    \mathcal{C}:= \{(x,u) \in \mathbb{R}^{n_x}{\times}\mathbb{R}^{n_u}\: \:|\: \: F x + G u \le b \},
\end{align}
where $F\in\mathbb{R}^{n_c\times n_x}, G\in\mathbb{R}^{n_c\times n_u}$ and $ b\in\mathbb{R}^{n_c}$. The control task is subject to a stage cost function of the form $l(x, u)$, whose cumulative sum has to be minimized over a (possibly) infinite horizon. Directly minimizing this cost is intractable, as it requires optimization over infinite number of variables and constraints. 



Instead, model predictive control (MPC) is used to find suboptimal input sequences \cite{borelli2017model}. In this approach, a receding horizon strategy is used  where the control inputs over the next $ N $  timesteps (called the prediction horizon) are optimized while also ensuring that the state after $ N $ timesteps reaches a terminal set $ \mathcal{X}_N:= \{x\: | \:Yx\le z\} $ which contains the origin. In this work, robustness is ensured using constraint tightening \cite{chisci2001systems,parsi2022computationally}. That is, the MPC optimization problem only uses a nominal model of the system, and tightens the constraints \eqref{eq:Constraints} to ensure robustness in closed loop. 

Thus, an optimization problem of the following form is solved at each time step $k$ using the available state measurement $ x_k $,
\begin{subequations}\label{eq:OnlOptProb}
    \begin{align}
        \min_{\hat{\mathbf{x}}_k,\hat{\mathbf{u}}_k}  \quad  l_N(\hat{x}_{N|k}) + \displaystyle\sum_{i=0}^{N-1} &l(\hat{x}_{i|k}, \hat{u}_{i|k}) \label{eq:OnlOptProb1} \\
		\text{s.t.} \quad A\hat{x}_{i|k} + B\hat{u}_{i|k} &= \hat{x}_{i+1|k}, \quad \hat{x}_{0|k}{=}x_k, \label{eq:OnlOptProb2}\\
            F\hat{x}_{i|k} + G\hat{u}_{i|k}  &\le b - t_i,  \quad i\in \mathbb{N}_{0}^{N-1}, \label{eq:OnlOptProb3} \\
            Y\hat{x}_{N|k} &\le z - t_N, \label{eq:OnlOptProb4}
    \end{align}
\end{subequations}
where $  \hat{\mathbf{x}}_k {:=} [\hat{x}_{0|k}\tr,\ldots,\hat{x}_{N|k}\tr]\tr$ and $\hat{\mathbf{u}}_k{:=} [\hat{u}_{0|k}\tr,\ldots,\hat{u}_{N-1|k}\tr]\tr $ represent the nominal state and input trajectories, as predicted by the nominal dynamics in \eqref{eq:OnlOptProb2}. In addition, \eqref{eq:OnlOptProb4} is a terminal constraint imposed on the last state in the prediction horizon, and $l_N(\cdot) $ is a terminal cost function. The terms $\{t_i\}_{i=0}^{N}$ are tightenings imposed on the constraints along the prediction horizon and the terminal set. The MPC components $l_N(\cdot) $, $Y$, $z$ and $\{t_i\}_{i=0}^{N}$ are design choices, and must be computed offline, i.e, before \eqref{eq:OnlOptProb} is solved. 

Let the optimal solution to \eqref{eq:OnlOptProb} be defined by the superscript $( ^*)$ on the optimization variables. The MPC control law is then defined as $\pi(x_k) =\hat{u}_{0|k}^*$. That is, only the first control input from $\hat{\mathbf{u}}_k^*$ is applied to the system, and \eqref{eq:OnlOptProb} is solved again at the following time step. 

In order to achieve desired closed-loop properties, relevant design criteria must be imposed on the MPC components $l_N(\cdot) $, $Y$, $z$ and $\{t_i\}_{i=0}^{N}$.  The following sections present methods formulating such design criteria.

\section{Constraint satisfaction using recursive feasibility}\label{Sec:RCS_RF}
Two main properties that are desired to be satisfied by robust MPC controllers are (i)  constraint satisfaction and (ii) recursive feasibility. Existing robust MPC methods are designed by explicitly accounting for both (i) and (ii), for example, in \cite{chisci2001systems}. In this work, a novel design strategy is proposed such that satisfying (ii) implicitly ensures that (i) is satisfied. For this purpose, the notion of recursive feasibility (RF) is formally presented \cite{lofberg2012oops}. 

\vspace{2pt}
\begin{Definition}[Recursive feasibility] \label{Def:RecFeas} 
An MPC controller is recursively feasible if, for any state $x_k$ such that the MPC optimization is feasible, the optimal control input computed by MPC when applied to the system results in a state $x_{k+1}$ which is also feasible for the MPC optimization problem.
\end{Definition}
In the context of the system \eqref{eq:Dynamics}-\eqref{eq:wBound} and the MPC \eqref{eq:OnlOptProb}, RF means that  for every $x_k$ such that \eqref{eq:OnlOptProb} is feasible, the MPC \eqref{eq:OnlOptProb} is also feasible for the state $x_{k+1}$ under the control law $\pi(x_k) = \hat{u}_{0|k}^*$ and for any realization of $w_k \in \mathcal{W}$ and $\Delta_k \in \mathcal{P}$. 

Note that $\pi(x_k)$ uses the optimal solution of \eqref{eq:OnlOptProb}, which cannot be easily expressed as an explicit function of $x_k$. Thus, RF is difficult to use as a design criterion or even verify \cite{lofberg2012oops}. In the light of these difficulties, it is useful to consider a stronger condition, pertaining to only the feasibility of MPC, which will ensure that RF is satisfied. 
\vspace{2pt}
\begin{Definition}[Strong recursive feasibility] \label{Def:RecFeas2} 
An MPC controller is strongly recursively feasible if, for any state $x_k$ such that MPC optimization is feasible, any feasible control input computed by MPC when applied to the system results in a state $x_{k+1}$ which is also feasible for the MPC optimization problem.
\end{Definition}
Note that strong recursive feasibility (SRF) is a sufficient condition for RF, and has been referred to in literature as \emph{strong recursive feasibility} \cite[Definition 2.2]{lofberg2012oops}, or \emph{robust strong feasibility} \cite[Definition 12]{kerrigan2001Robust}. Because SRF does not depend on the optimality of the control input, it can be used in the design of MPC controllers. Although using SRF for design could lead to conservatism, it must be noted that  most existing robust MPC methods satisfy SRF \cite{parsi2022scalable,schwenkel2022}. 

In the context of the system \eqref{eq:Dynamics}-\eqref{eq:wBound} and the MPC \eqref{eq:OnlOptProb}, SRF can be formalized as follows. Let $\mathcal{U}(x_k)$ be the set of input sequences such that \eqref{eq:OnlOptProb2}-\eqref{eq:OnlOptProb4} are satisfied for a given state $x_k$. Additionally, define $\mathcal{U}_0(x_k)$  as the set of feasible inputs at the first predicted time step, that is,  
\begin{align}
\mathcal{U}_0(x_k) := \left\{\hat{u}_{0|k}\: \:|\:\:  \exists \{\hat{u}_{l|k}\}_{l=0}^{N-1} \in \mathcal U(x_k) \right\}.
\end{align}
The MPC \eqref{eq:OnlOptProb} is said to be SRF if, for every state $x_k$ such that \eqref{eq:OnlOptProb} is feasible, under any control input $u_k \in \mathcal{U}_0(x_k)$,  and for any realization of $w_k \in \mathcal{W}$ and $\Delta_k \in \mathcal{P}$,  the optimization \eqref{eq:OnlOptProb} is also feasible for the state $x_{k+1}$ following  the dynamics \eqref{eq:Dynamics}.

In this paper, we will rely on the fact that, under very mild conditions, RF implies closed-loop constraint satisfaction. That is, if the MPC scheme is recursively feasible, then the conditions under which constraints \eqref{eq:Constraints} are always satisfied by the closed-loop trajectories are trivial.  These statements are formalized for the closed loop \eqref{eq:Dynamics}-\eqref{eq:OnlOptProb} in the following theorem.
\vspace{2pt}
\begin{Theorem}\label{Th:RF_RCS}
    Let the optimization \eqref{eq:OnlOptProb} be recursively feasible and also be feasible at time $k=0$. In addition, let $t_0 \ge 0$. Then, the closed-loop trajectories resulting from the dynamics \eqref{eq:Dynamics} and MPC control policy $\pi(x_k)=\hat{u}_{0|0}^*$ never violate the constraints \eqref{eq:Constraints}.  
\end{Theorem}
\begin{proof}
Because the optimization \eqref{eq:OnlOptProb} is feasible at time step $k{=}0$ and $t_0 {\ge} 0$, the applied input $u_0 = \hat{u}_{0|0}^*$  satisfies  
\begin{align}\label{eq:RecFeasproof}
    F x_0 + G u_0 \le b - t_0\le b.
\end{align}
Therefore, \eqref{eq:Constraints} is satisfied at time $k{=}0$. Moreover, because RF is satisfied, 
 the optimization \eqref{eq:OnlOptProb} always remains feasible under the MPC control law $\pi(x)$. By using a similar argument as the one presented above for $k{=}0$,  \eqref{eq:Constraints} is satisfied for all $k{\ge} 0$.
\end{proof}

It can be useful to stress here that Theorem \ref{Th:RF_RCS} -- beyond RF -- requires only that problem \eqref{eq:OnlOptProb} is initially feasible, and that the initial constraint \eqref{eq:OnlOptProb3} is ``properly" tightened, i.e. that $t_0{\geq} 0$. The latter requirement is arguably trivial.

\begin{Remark}
Theorem \ref{Th:RF_RCS} relies only on the feasibility of the solutions to \eqref{eq:OnlOptProb} and not on their optimality. As a consequence, a straightforward corollary is that Theorem \ref{Th:RF_RCS} also holds under SRF. Hence, SRF of the closed loop \eqref{eq:Dynamics}-\eqref{eq:OnlOptProb} will guarantee constraint satisfaction under any control input $u_k \in \mathcal{U}_0(x_k)$. That is, even if problem \eqref{eq:OnlOptProb} cannot be solved to optimality, constraint satisfaction is guaranteed by SRF.
\end{Remark}

Most robust MPC methods explicitly build outer bounds on worst-case trajectories of the system to ensure constraint satisfaction \cite{chisci2001systems,fleming2014robust}, which can be interpreted as open-loop approaches to robust MPC design. This is because trajectories the predicted by MPC in open loop are designed to lie inside the system constraints even with perturbations acting along the prediction horizon. Theorem \ref{Th:RF_RCS} enables a design which circumvents this construction altogether, and is the key element of the closed-loop design presented in this work. 

In the following sections, the MPC components $Y, z, \{t_i\}_{i=0}^{N}$ and $l_N(\cdot)$ will be designed such that SRF is satisfied. The objective is formulate the design of these components as an optimization problem that needs to be solved only once, offline. This offline optimization will be presented in  Section \ref{Sec:RobustMPCDesign}.

\section{Necessary and sufficient conditions for SRF}
In order to utilize Theorem \ref{Th:RF_RCS} for design, the SRF of \eqref{eq:OnlOptProb} needs to be guaranteed. In this section, necessary and sufficient conditions for SRF are derived. The constraints \eqref{eq:OnlOptProb2}-\eqref{eq:OnlOptProb4} can be compactly written as 
\begin{align}\label{eq:FeasCond_k}
    \mathbf{H}_{xu}\mathbf{S} \mathbf{s}_k \le \mathbf{b}-\mathbf{t},
\end{align}
where
\begin{align}\label{eq:Notation1}
&\mathbf{s}_k = \Big[x_k\tr~~\hat{\mathbf{u}}_{k}\tr\Big]\tr, ~\mathbf{H}_{xu}= 
\begin{bmatrix}
    I_N \otimes F & 0 & I_N \otimes G\\
    0 & Y & 0 
\end{bmatrix},
\nonumber\\
&\mathbf{S} =\left[ \begin{array}{c|cccc} 
    I & 0 & ... & ... & 0 \\ A & B & 0 & ... & 0  \\ \vdots & \vdots & \ddots & \ddots & \vdots \\ A^{N-1} & A^{N-2}B & ... & B & 0 \\ A^N & A^{N-1} B &...& AB &  B \\ 0 & \multicolumn{4}{c}{I_{Nn_u}} 
    \end{array}    \right]
    = \left[ \begin{array}{c|c} 
        \mathbf{S}_x &  \mathbf{S}_u 
    \end{array}    \right]
    \nonumber\\
&\mathbf{b} = \Big[\textbf{1}_N\tr\otimes b\tr~~z\tr\Big]\tr, ~\text{and }\mathbf{t} = [t_0~~t_1~~\hdots~~t_N]\tr.
\end{align}
One can readily observe that the vector $\mathbf{s}_k$ collects the initial state $x_k$ and input trajectory $\hat{\mathbf{u}}_k$. The matrix $\mathbf{S}$ maps the initial state $x_k$ and input trajectory $\hat{\mathbf{u}}_k$, i.e., $\mathbf{s}_k$, to the predicted \mbox{state-input} trajectories; more precisely, $\mathbf{Ss}_k = [\hat{\mathbf{x}}_k\tr~~\hat{\mathbf{u}}_k\tr]\tr$. The matrix $\mathbf{H}_{xu}$ maps the predicted state-input trajectories ($\mathbf{Ss}_k$) into the left-hand side of constraints \eqref{eq:OnlOptProb3}-\eqref{eq:OnlOptProb4}. Vectors $\mathbf{b}$ and $\mathbf{t}$ collect the right-hand side of constraints \eqref{eq:OnlOptProb3}-\eqref{eq:OnlOptProb4}. Let $\mathbb{P}_0$ denote the set of all feasible vectors $\mathbf{s}_k$ satisfying \eqref{eq:FeasCond_k}, that is,
\begin{align} \label{eq:Poytope_P0}
    \mathbb{P}_0 &:= \Set{\mathbf{s} \in \mathbb{R}^{n_x{+}N n_u}}{ \mathbf{H}_{xu}\mathbf{S} \mathbf{s} \le \mathbf{b}-\mathbf{t}}. 
\end{align}
Observe that SRF is satisfied if, for every $\mathbf{s}_k\in\mathbb{P}_0$, and for all possible realizations of $x_{k+1}$ under the control input $u_k = \hat{u}_{0|k}$, disturbance $w_k\in \mathcal{W}$ and perturbation $\Delta_k\in \mathcal{P}$, there exists an input sequence $\hat{\mathbf{u}}_{k+1}$ satisfying 
\begin{align}\label{eq:FeasCond_k+1}
    \mathbf{H}_{xu}\mathbf{S} \mathbf{s}_{k+1} \le \mathbf{b}-\mathbf{t},
\end{align}
where $\mathbf{s}_{k+1}=[x_{k+1}\tr~\hat{\mathbf{u}}_{k+1}\tr]\tr$. Note that (\ref{eq:FeasCond_k+1}) can be expressed in terms of $\mathbf{s}_{k}$, $w_k$, and $\Delta_k$ by re-writing the state-input trajectories $\mathbf{Ss}_{k+1}$ as
\begin{align}\label{eq:s_k+1}
    \mathbf{S} \mathbf{s}_{k+1} &=   \mathbf{S}_x x_{k+1} +  \mathbf{S}_u \hat{\mathbf{u}}_{k+1}, \nonumber\\
    &=  \mathbf{S}_x (A x_k + B u_k)  + 
   \mathbf{S}_x\big( B_w w_k +B_p \Delta_k (D_x x_k \nonumber\\ 
   &\quad + D_u u_k + D_w w_k)\big)+\mathbf{S}_u \hat{\mathbf{u}}_{k+1} \nonumber\\
       & =(\mathbf{L} {+} \mathbf{S}_x B_p \Delta_k  D_{xu}  ) \mathbf{s}_k \nonumber\\ 
     &\quad + \mathbf{S}_x (B_w {+} B_p \Delta_k D_w )w_k +  \mathbf{S}_u \hat{\mathbf{u}}_{k+1} \nonumber\\
     &=\mathbf{C_ss}_k + \mathbf{C_w}w_k + \mathbf{S}_u \hat{\mathbf{u}}_{k+1},
\end{align}
where
\begin{align}
&\mathbf{C}_{\mathbf{s}} = \mathbf{L} + \mathbf{S}_x B_p \Delta_k  D_{xu}, \: \: \mathbf{C}_{\mathbf{w}} = \mathbf{S}_x (B_w  + B_p \Delta_k D_w), \nonumber \\
&\mathbf{L} = \begin{bmatrix} 
    A & B & 0_{n_x,(N-1)n_u}   \\ A^2 & AB & 0_{n_x,(N-1)n_u} \\ \vdots & \vdots & \vdots \\ A^{N+1} & A^N B & 0_{n_x,(N-1)n_u} \\  \multicolumn{3}{c}{0_{Nn_u,n_x+Nn_u}} 
    \end{bmatrix}, 
    D_{xu} = \begin{bmatrix} D_x\\
    D_u
    \\
    0_{n_p,(N-1)n_u}
    \end{bmatrix}.
\end{align}
Note that the matrix $\mathbf{C}_{\mathbf{s}}$ maps $\mathbf{s}_k$ into the predicted state-input trajectories at time $k+1$. Similarly, $\mathbf{C}_{\mathbf{w}}$ provides the gain from disturbance $w_k$ at time $k$ to the predicted state-input trajectories at time $k+1$.

Using the above definitions, necessary and sufficient conditions for SRF are presented in the following theorem. 
\begin{Theorem}\label{Thm:NecSuff}
The closed loop system \eqref{eq:Dynamics}-\eqref{eq:OnlOptProb} satisfies SRF  if, and only if, for all $ \mathbf{s}_k {\in} \mathbb{P}_0$, $w_k{\in}\mathcal{W}$, and $ j{\in} \mathbb{N}_{1}^{n_\Delta}$, there exist $ \hat{\mathbf{u}}_{k+1}^j $ such that 
\begin{align}\label{eq:NecSuff}
     \mathbf{H}_{xu}  \left(\mathbf{C}^j_{\mathbf{s}} \mathbf{s}_k + \mathbf{C}^j_{\mathbf{w}} w_k + \mathbf{S_u}\hat{\mathbf{u}}_{k+1}^j \right)  \leq \mathbf{b}-\mathbf{t}, &
\end{align}
where $\mathbf{C}^j_{\mathbf{s}} = \mathbf{L} + \mathbf{S}_x B_p \Delta^j  D_{xu}$, and  $\mathbf{C}^j_{\mathbf{w}} = \mathbf{S}_x (B_w  + B_p \Delta^j D_w)$.
\end{Theorem}
\begin{proof}
The idea of the proof is briefly described first. By definition, SRF requires that  for every $\mathbf{s}_k {\in} \mathbb{P}_0$, $w_k {\in} \mathcal{W}$ and $\Delta_k {\in} \mathcal{P}$, there exists an input sequence $\hat{\mathbf{u}}_{k+1}$ satisfying \eqref{eq:FeasCond_k+1}. Note that the qualifier for all $\Delta_k \in \mathcal{P}$ can be rewritten as for all $ j {\in} \mathbb{N}_{1}^{n_\Delta} $, with $\Delta_k{=}\Delta^j$. Thus, existence of independent input sequences $\hat{\mathbf{u}}_{k+1}^j$, each corresponding to $\Delta_k{=}\Delta^j$, ensures SRF. Using \eqref{eq:FeasCond_k+1} and \eqref{eq:s_k+1}, this condition can be written as \eqref{eq:NecSuff}. The detailed proof follows.
\\[1mm]
\textit{Sufficiency:}\\
Let there exist inputs $ \hat{\mathbf{u}}_{k+1}^j $  such that \eqref{eq:NecSuff} is satisfied for all $ \mathbf{s}_k {\in} \mathbb{P}_0$, $w_k{\in}\mathcal{W}$, and $ j{\in} \mathbb{N}_{1}^{n_\Delta}$. Moreover, let the true perturbation $\Delta_k$ be given as 
\begin{align}\label{eq:Delta_k}
    \Delta_k = \sum_{j=1}^{n_\Delta} \tau_j \Delta^j, 
\end{align}
where $ \sum_{j=1}^{n_\Delta} \tau_j = 1$ and $\tau_j \ge 0$. Such a representation is possible as $\Delta_k$ lies inside $\mathcal{P}$ defined in \eqref{eq:DeltaBound}. Consider the input candidate sequence $\hat{\mathbf{u}}_{k+1} =  \sum_{j=1}^{n_\Delta} \tau_j \hat{\mathbf{u}}_{k+1}^j$ for verifying  \eqref{eq:FeasCond_k+1}. Then, using \eqref{eq:s_k+1}, 
\begin{align}\label{eq:Suff_NecSuff}
    &\mathbf{S} \mathbf{s}_{k+1}  \nonumber\\
    &=(\mathbf{L} {+} \mathbf{S}_x B_p \Delta_k  D_{xu}  ) \mathbf{s}_k  {+}  
     \mathbf{S}_u \hat{\mathbf{u}}_{k{+}1} {+} \mathbf{S}_x (B_w {+} B_p \Delta_k D_w )w_k \nonumber\\
     &= (\mathbf{L} {+} \mathbf{S}_x B_p \sum_{j=1}^{n_\Delta} \tau_j \Delta^j  D_{xu}  ) \mathbf{s}_k  {+}  
     \mathbf{S}_u \sum_{j=1}^{n_\Delta}\tau_j  \hat{\mathbf{u}}_{k{+}1}^j {+} \nonumber\\
     & \qquad \qquad \mathbf{S}_x (B_w {+} B_p \sum_{j=1}^{n_\Delta} \tau_j \Delta^j D_w )w_k \nonumber\\
     &= \sum_{j=1}^{n_\Delta}\tau_j \big( (\mathbf{L} {+} \mathbf{S}_x B_p  \Delta^j  D_{xu}  ) \mathbf{s}_k  {+}   \mathbf{S}_u \hat{\mathbf{u}}_{k{+}1}^j {+}   \nonumber\\
     & \qquad \qquad \mathbf{S}_x (B_w {+} B_p \Delta^j D_w )w_k \big)\nonumber\\
     &= \sum_{j=1}^{n_\Delta}\tau_j \left(\mathbf{C}^j_{\mathbf{s}} \mathbf{s}_k + \mathbf{S_u}\hat{\mathbf{u}}_{k+1}^j  + \mathbf{C}^j_{\mathbf{w}} w_k \right).
\end{align}
Therefore, using \eqref{eq:NecSuff} and \eqref{eq:Suff_NecSuff},  
\begin{align}
   & \mathbf{H}_{xu}  \mathbf{S} \mathbf{s}_{k+1}  =  \mathbf{H}_{xu} \sum_{j=1}^{n_\Delta}\tau_j \left(\mathbf{C}^j_{\mathbf{s}} \mathbf{s}_k + \mathbf{S_u}\hat{\mathbf{u}}_{k+1}^j  + \mathbf{C}^j_{\mathbf{w}} w_k \right) \nonumber\\
    & = \sum_{j=1}^{n_\Delta}\tau_j \mathbf{H}_{xu}  \left(\mathbf{C}^j_{\mathbf{s}} \mathbf{s}_k + \mathbf{S_u}\hat{\mathbf{u}}_{k+1}^j  + \mathbf{C}^j_{\mathbf{w}} w_k \right) \nonumber\\
    &\le \sum_{j=1}^{n_\Delta}\tau_j (\mathbf{b}-\mathbf{t})  = \mathbf{b}-\mathbf{t}, 
\end{align}
where the inequality is obtained as \eqref{eq:NecSuff} is satisfied by the inputs $\hat{\mathbf{u}}_{k+1}^j$.
Thus, \eqref{eq:FeasCond_k+1} is satisfied for all $\Delta_k{\in} \mathcal{P}$, $\mathbf{s}_k {\in }\mathbb{P}_0$, and $w_k {\in} \mathcal{W}$, and the closed-loop satisfies SRF. 

\noindent\textit{Necessity: }\\
Let the closed-loop \eqref{eq:Dynamics}-\eqref{eq:OnlOptProb} satisfy SRF. That is, for all $\mathbf{s}_k {\in} \mathbb{P}_0$, $\Delta_k{\in} \mathcal{P}$ and $w_k {\in} \mathcal{W}$, there exists $\hat{\mathbf{u}}_{k+1}$ satisfying \eqref{eq:FeasCond_k+1}. Then, consider that for each $j \in \mathbb{N}_{1}^{n_\Delta}$, $\hat{\mathbf{u}}_{k+1}^j$ represents the feasible input sequences for \eqref{eq:FeasCond_k+1} when  $\Delta_k{=}\Delta^j$, and for $\mathbf{s}_k{\in }\mathbb{P}_0 $, $w_k {\in} \mathcal{W}$. Using \eqref{eq:FeasCond_k+1} and \eqref{eq:s_k+1}, one obtains
\begin{align}
    \mathbf{H}_{xu}  (\mathbf{L} {+} \mathbf{S}_x B_p \Delta^j  D_{xu}  ) \mathbf{s}_k  +\mathbf{S}_u \hat{\mathbf{u}}_{k{+}1}^j + \quad & \nonumber \\
      \mathbf{S}_x (B_w {+} B_p \Delta^j D_w )w_k &\le \mathbf{b}-\mathbf{t} \nonumber\\
     \implies \mathbf{H}_{xu}  \left(\mathbf{C}^j_{\mathbf{s}} \mathbf{s}_k + \mathbf{C}^j_{\mathbf{w}} w_k + \mathbf{S_u}\hat{\mathbf{u}}_{k+1}^j \right) &\le \mathbf{b}-\mathbf{t},
\end{align}
proving that \eqref{eq:NecSuff} is satisfied for all $j {\in} \mathbb{N}_{1}^{n_\Delta}$.
\end{proof}
Although Theorem 2 provides necessary and sufficient conditions for the characterization of SRF, \eqref{eq:NecSuff} is difficult to directly use in the design of the MPC components
$Y$,$z$ and $\mathbf{t}$. The difficulty arises as one must guarantee the existence of input sequences satisfying \eqref{eq:NecSuff} for each $\mathbf{s}_k \in \mathbb{P}_0$ and $w_k \in \mathcal{W}$.

Theorem \ref{Thm:NecSuff} explicitly captures the conditions under which SRF holds. Whereas Theorem \ref{Thm:NecSuff} formulates these conditions as existence of input sequences, the criteria can be written as set inclusions using ideas from set-invariance \cite{blanchini2008set}. The early work \cite{kerrigan2001Robust} presents necessary and sufficient conditions to achieve SRF when uncertainty occurs as additive disturbances. However, the resulting conditions are not suited for design. Recently, analogous conditions were formulated for multiplicative perturbations in \cite{abdelsalam2021synthesis}, where polytope projections were used to verify SRF, given the MPC parameters. 

\section{Sufficient conditions for SRF using feedback}\label{Sec:SuffCond_Feedback}
Theorem \ref{Thm:NecSuff} is difficult to directly use as a design condition because of the requirement to guarantee the existence of input sequences. In this section, such a guarantee is achieved by parameterizing the inputs using feedback gains.

Note that Theorem \ref{Thm:NecSuff} allows one to design independent input sequences $\hat{\mathbf{u}}_{k+1}^j $  for each $\Delta^j$. This is a powerful feature, and differs from most existing open-loop strategies, where a single feedback law is designed to compensate for all perturbations and disturbances \cite{parsi2022scalable}.  Exploiting the fact that $n_\Delta$ independent input sequences are sufficient to ensure SRF,  for each $j\in \mathbb{N}_{1}^{n_\Delta}$, the input sequences $\hat{\mathbf{u}}_{k+1}^j $ are parameterized as
\begin{align}\label{eq:InputParameterization}
    \hat{u}_{i|k+1}^j &= \hat{u}_{i+1|k}^j + M_i^j B_w w_k + K^j_{i,\Delta} y_k, \: \forall i \in \mathbb{N}_{0}^{N-2}, \nonumber\\
    \hat{u}_{N-1|k+1}^j &= K^j \hat{x}_{N|k} + M_{N{-}1}^j B_w w_k  + K^j_{N{-}1,\Delta} y_k, 
\end{align}
where $y_k = \begin{bmatrix}
x_k\tr & u_k\tr \end{bmatrix}\tr$, and for $j\in \mathbb{N}_{1}^{n_\Delta}$ and $i \in \mathbb{N}_{0}^{N-
1}$,  $K^j$ represents a terminal state feedback gain, $M_i^j$ represents a disturbance feedback gain and $K_{i,\Delta}^j$ represents a perturbation feedback gain. The use of disturbance feedback gains is common in robust MPC \cite{goulart2006optimization}. Moreover, each $K_{i,\Delta}^j$ compensates for the effect of $\Delta^j$ on the nominal trajectory. The term $y_k$ is used for feedback as the effect of the perturbation $\Delta_k$  on the nominal trajectories also depends on $y_k$ (see \eqref{eq:s_k+1}).

Condition \eqref{eq:NecSuff} can now be rewritten as the existence of suitable feedback gains, which can be solved for offline. The following notation is introduced to account for the newly introduced feedback gains. Let for all $j\in\mathbb{N}_{1}^{n_\Delta}$
\begin{align}
    A_K^j &:= A + B K^j, \: 
    \mathbf{M}^j := \begin{bmatrix} {M_0^j}\tr & {M_1^j}\tr & \hdots & {M_{N{-}1}^j}\tr\end{bmatrix}\tr,\nonumber\\
    \mathbf{K}^j_{\Delta} &:= \begin{bmatrix}  {K^j_{0,\Delta}}\tr & {K^j_{1,\Delta}}\tr & \hdots & {K^j_{N-1,\Delta}}\tr   \end{bmatrix} \tr,\nonumber\\
    \mathbf{L}_K^j &:= \begin{bmatrix}
    A & B & 0  & ... & 0 \\ A^2 & AB & B & ... & 0 \\ \vdots & \vdots & \vdots & \ddots & \vdots \\ A^N & A^{N-1} B & A^{N-2} B &... &  B \\ A_K^j A^N & A_K^j A^{N-1} B & A_K^j A^{N-2} B &... &  A_K^j B \\[0.2cm] 
    0 & 0 & \multicolumn{3}{c}{\quad \quad I_{(N-1)n_u}}  \\[0.2cm] 
    K^j A^N & K^j A^{N-1} B & K^j A^{N-2} B &... &  K^j B
    \end{bmatrix}\nonumber, \\
    \mathbf{C}_K^j &:= \mathbf{L}_K^j + \mathbf{S}_x B_p \Delta^j  D_{xu} + \nonumber\\
     &\qquad \quad \mathbf{S}_u \mathbf{K}^j_{\Delta} \begin{bmatrix}
        I_{n_x+n_u}  & 0_{n_x+n_u, n_u(N-1)} 
    \end{bmatrix}\nonumber, \\
    \mathbf{C}_M^j &:= \mathbf{S}_x (B_w  + B_p \Delta^j D_w )+ \mathbf{S}_u\mathbf{M}^jB_w.
\end{align}
Here matrix $\mathbf{L}_K^j$ maps a pair of initial state and the nominal input trajectory (that is, $\mathbf{s}_k$) into the corresponding nominal predicted state-input trajectory at the next time step, under the parameterization \eqref{eq:InputParameterization}. Matrix $\mathbf{C}_K^j$  maps  $\mathbf{s}_k$ into a perturbed state-input trajectory at the next time step, when the perturbation  $\Delta_k= \Delta^j$  and the input is parameterized as \eqref{eq:InputParameterization}. Similarly, $\mathbf{C}_M^j$ maps the effects of a disturbance $w_k$ on the state-input trajectory at the next time step. 

Consider the following polytopes in $\mathbb{R}^{n_x+Nn_u+n_w}$
\begin{align}
    \mathbb{P}_1 &:= \Set*{\begin{bmatrix}
        \mathbf{s} \\ w
    \end{bmatrix} \: }{\: \begin{bmatrix}
        \mathbf{H}_{xu}\mathbf{S} & 0 \\
        0 & H_w
    \end{bmatrix}
    \begin{bmatrix}
        \mathbf{s} \\ w
    \end{bmatrix} {\leq}
    \begin{bmatrix}
        \mathbf{b}-\mathbf{t} \\ h_w
    \end{bmatrix} }, 
    \nonumber \\ 
    \text{and }
    \mathbb{P}_2^j &:= \Set*{\begin{bmatrix}
        \mathbf{s} \\ w
    \end{bmatrix} }{\mathbf{H}_{xu}\begin{bmatrix}
        \mathbf{C}_K^j &   \mathbf{C}_M^j
    \end{bmatrix}
    \begin{bmatrix}
        \mathbf{s} \\ w
    \end{bmatrix} \leq \mathbf{b}-\mathbf{t}  }, \label{eq:Poytope_P1_P2j} 
\end{align}
for $j\in \mathbb{N}_{1}^{n_\Delta}$. Polytopes $\mathbb{P}_1$ and $\mathbb{P}_2^j$ are depicted in Figure \ref{fig:CLR4}, and can be interpreted in the following manner. Recall that $\mathbb{P}_0$ is the set of feasible states and input sequences for MPC, as defined in \eqref{eq:Poytope_P0}. Then, $\mathbb{P}_1$ is the augmentation of $\mathbb{P}_0$ with the set of possible disturbances \eqref{eq:wBound}. Moreover, $\mathbb{P}_2^j$ is the set of states, inputs and disturbances acting at time step $k$ such that, if $\Delta_{k} {=} \Delta^j$, $\mathbf{s}_{k+1}$ lies inside $\mathbb{P}_0$. That is, $\mathbb{P}_2^j$ is the pre-image of $\mathbb{P}_0$ under the dynamics \eqref{eq:Dynamics} with $\Delta_{k} {=} \Delta^j$, $w_k{\in} \mathcal{W}$ and $u_k$ as in \eqref{eq:InputParameterization}. 

Therefore, it can be seen that, if $\mathbb{P}_1\subseteq \mathbb{P}_2^j$ and $\Delta_{k} {=} \Delta^j$, for any $s_k\in \mathbb{P}_0$ and $w_k\in\mathcal{W}$, $s_{k+1}$ lies inside $\mathbb{P}_0$. Using this approach, polytopic set inclusions can be defined to guarantee SRF as shown in the following theorem.

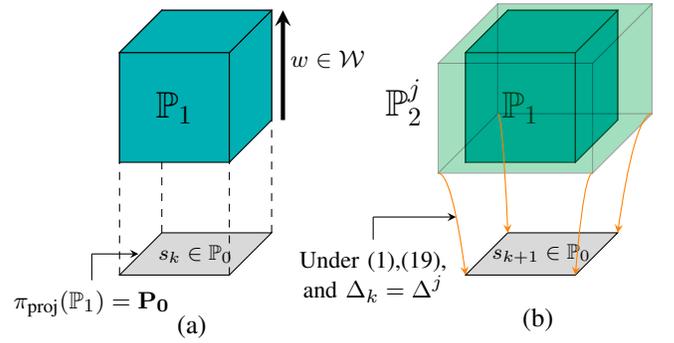
\begin{figure}
\begin{tikzpicture}[font=\small]
    \begin{scope}[yslant=0,shift={(-1.,0)}]
    \tikzmath{\a1=1.46; \b1=\a1/2; \d1=1.5; \d2=1.2;}
        \coordinate (sk0) at (0,-\d1,0);
    	\coordinate (sk1) at (\a1,-\d1,0);
    	\coordinate (sk2) at (\a1,-\d1,\a1);
    	\coordinate (sk3) at (0,-\d1,\a1);
    	\filldraw[fill=black!15] (sk0)--(sk1)--(sk2)--(sk3)--cycle;
    	\draw[black,dashed] (sk3) -- (0,0,\a1);
    	\draw[black,dashed] (sk0) -- (0,0,0);
    	\draw[black,dashed] (sk1) -- (\a1,0,0);
    	\draw[black,dashed] (sk2) -- (\a1,0,\a1);
        \filldraw[fill=TealBlue] (0,0,0) -- (0,0,\a1) -- (0,\a1,\a1) -- (0,\a1,0) -- cycle;
        \filldraw[fill=TealBlue] (0,0,0) -- (0,\a1,0) -- (\a1,\a1,0) -- (\a1,0,0) -- cycle;
    	\filldraw[fill=TealBlue] (0,0,\a1) -- (0,\a1,\a1) -- (\a1,\a1,\a1) -- (\a1,0,\a1) -- cycle;
    	\filldraw[fill=TealBlue] (\a1,0,0) -- (\a1,0,\a1) -- (\a1,\a1,\a1) -- (\a1,\a1,0) -- cycle;
    	\filldraw[fill=TealBlue] (0,\a1,\a1) -- (\a1,\a1,\a1) -- (\a1,\a1,0) -- (0,\a1,0) -- cycle;
    	\node at (\b1,\b1,\a1) {\Large$\mathbb{P}_1$};
    	\node at (\b1,-\d1,\b1) {\footnotesize$s_k\in\mathbb{P}_0$};
    	\draw[ultra thick,-stealth] (\a1+0.15,0,0) -- (\a1+0.15,\a1,0);
    	\node[anchor=west] at (\a1+0.13,\b1+0.0725,0) {$w\in\mathcal{W}$};
    	\coordinate (p1) at (-0.03,-\d1,\b1);
    	\coordinate (p2) at (-0.65,-\d1,\b1);
    	\coordinate (p3) at (-0.65,-\d2-0.65,\b1);
    	\draw[-stealth] (p2)--(p1); \draw (p3)--(p2);
    	\node[anchor=north] at (p3) {\begin{tabular}{c}$\pi_{\text{proj}}(\mathbb{P}_1)=\mathbf{P_0}$\end{tabular}};
    	\node at (0.4,-2.75) {\normalsize (a)};
    \end{scope}
    \begin{scope}[yslant=0,shift={(3.6,0,0)}]
        \tikzmath{\a1=1.46; \b1=\a1/2;\a2=1.45;\s2=0.4; \d1=1.5; \d2=1.2;}
        \coordinate (sk0) at (0,-\d1,0);
    	\coordinate (sk1) at (\a1,-\d1,0);
    	\coordinate (sk2) at (\a1,-\d1,\a1);
    	\coordinate (sk3) at (0,-\d1,\a1);
    	\filldraw[fill=black!15] (sk0)--(sk1)--(sk2)--(sk3)--cycle;
        \filldraw[fill=TealBlue] (0,0,0) -- (0,0,\a1) -- (0,\a1,\a1) -- (0,\a1,0) -- cycle;
        \filldraw[fill=TealBlue] (0,0,0) -- (0,\a1,0) -- (\a1,\a1,0) -- (\a1,0,0) -- cycle;
    	\filldraw[fill=TealBlue] (0,0,\a1) -- (0,\a1,\a1) -- (\a1,\a1,\a1) -- (\a1,0,\a1) -- cycle;
    	\filldraw[fill=TealBlue] (\a1,0,0) -- (\a1,0,\a1) -- (\a1,\a1,\a1) -- (\a1,\a1,0) -- cycle;
    	\filldraw[fill=TealBlue] (0,\a1,\a1) -- (\a1,\a1,\a1) -- (\a1,\a1,0) -- (0,\a1,0) -- cycle;
    	\node at (\b1,\b1,\a1) {\Large$\mathbb{P}_1$};
    	\node at (\b1,-\d1,\b1) {\footnotesize$s_{k+1}\in\mathbb{P}_0$};
    	
    	\tikzmath{\a11=1.46; \b1=\a11/2;\a22=2.05;\s22=0.2; \d11=1.5;\a33=2.05; \d22=1.2;}
        \coordinate (ss0) at (-\s22,0,-\s22);
        \coordinate (ss1) at (\a22-\s22,0,-\s22);
        \coordinate (ss22) at (\a22-\s22,0,\a33-\s22);
        \coordinate (ss3) at (-\s22,0,\a33-\s22);
        \coordinate (label) at (-\s22,\a22/2,\a22);
        \node[anchor=east] at (label) {\Large$\mathbb{P}_2^j$};
        \filldraw[fill=Green,opacity=0.2] (-\s22,0,-\s22) -- (-\s22,0,\a33-\s22) -- (-\s22,\a11,\a33-\s22) -- (-\s22,\a11,-\s22) -- cycle;
        \filldraw[fill=Green,opacity=0.2] (-\s22,0,-\s22) -- (-\s22,\a11,-\s22) -- (\a22-\s22,\a11,-\s22) -- (\a22-\s22,0,-\s22) -- cycle;
    	\filldraw[fill=Green,opacity=0.2] (-\s22,0,\a33-\s22) -- (-\s22,\a11,\a33-\s22) -- (\a22-\s22,\a11,\a33-\s22) -- (\a22-\s22,0,\a33-\s22) -- cycle;
    	\filldraw[fill=Green,opacity=0.2] (\a22-\s22,0,-\s22) -- (\a22-\s22,0,\a33-\s22) -- (\a22-\s22,\a11,\a33-\s22) -- (\a22-\s22,\a11,-\s22) -- cycle;
    	\filldraw[fill=Green,opacity=0.2] (-\s22,\a11,\a33-\s22) -- (\a22-\s22,\a11,\a33-\s22) -- (\a22-\s22,\a11,-\s22) -- (-\s22,\a11,-\s22) -- cycle;
    	\filldraw[fill=Green,opacity=0.2] (-\s22,0,-\s22) -- (\a22-\s22,0,-\s22) -- (\a22-\s22,0,\a33-\s22) -- (-\s22,0,\a33-\s22) -- cycle;
    	\draw[orange,-stealth] (ss0) parabola (sk0);
    	\draw[orange,-stealth] (ss1) parabola (sk1);
    	\draw[orange,-stealth] (ss22) parabola (sk2);
    	\draw[orange,-stealth] (ss3) parabola (sk3);
    	\coordinate (l1) at (0.1,-0.5,\a22);
    	\coordinate (l2) at (-1,-0.5,\a22);
    	\coordinate (l3) at (-1,-0.8,\a22);
    	\draw[-stealth] (l2) -- (l1); \draw (l3)--(l2);
    	\node[anchor=north] at (l3) {\begin{tabular}{c}Under \eqref{eq:Dynamics},\eqref{eq:InputParameterization},\\ and $\Delta_k=\Delta^j$\end{tabular}};
    	
        \node at (0.4,-2.65) {\normalsize (b)};
    \end{scope}
\end{tikzpicture}
\caption{(a) Projection of $\mathbb{P}_1$ on the space of $s_k$ overlaps with $\mathbb{P}_0$ by definition. (b) $\mathbb{P}_2^j$ contains all $(s_k,w_k)$ such that $s_{k+1} \in \mathbb{P}_0$  holds under the dynamics \eqref{eq:Dynamics}, with $\Delta_k=\Delta^j$, $w_k\in\mathcal{W}$  and input parameterization \eqref{eq:InputParameterization}. }\label{fig:CLR4}
\end{figure}

\begin{Theorem}\label{Thm:SuffCond}
    Consider the polytopes $\mathbb{P}_1$ and $\{\mathbb{P}_2^j\}_{j=1}^{n_{\Delta}}$  proposed in \eqref{eq:Poytope_P1_P2j}. Then, the closed loop \eqref{eq:Dynamics}-\eqref{eq:OnlOptProb} satisfies SRF if there exist feedback gains $\{K^j\}_{j=1}^{n_\Delta}$, $\{ \{M_i^j\}_{j=i}^{n_{\Delta}} \}_{i=0}^{N-1}$ and $\{ \{ K_{i,\Delta}^j\}_{j=i}^{n_{\Delta}}\}_{i=0}^{N-1} $  such that
\begin{align}\label{eq:SuffCond}
     \mathbb{P}_1 \subseteq \mathbb{P}_2^j, \quad \forall j \in \mathbb{N}_1^{n_{\Delta}}.
\end{align}
\end{Theorem}
\begin{proof}
For \eqref{eq:OnlOptProb} to be strongly recursively feasible, \eqref{eq:FeasCond_k+1} must be satisfied for every $\begin{bmatrix}
        \mathbf{s}_k \tr & w_k\tr
    \end{bmatrix}\tr \in \mathbb{P}_1$ and $\Delta_k \in \mathcal{P}$. 

Let $\Delta_k $ be defined as in \eqref{eq:Delta_k} and a candidate control input sequence be defined as $\hat{\mathbf{u}}_{k+1} =  \sum_{j=1}^{n_\Delta} \tau_j \hat{\mathbf{u}}_{k+1}^j$, where $\hat{\mathbf{u}}_{k+1}^j$ are parameterized according to \eqref{eq:InputParameterization}. 
Substituting the candidate sequence in \eqref{eq:s_k+1}, 
\begin{subequations}\label{eq:SuffCondPf1}
\begin{align}
    &\mathbf{S} \mathbf{s}_{k+1} = (\mathbf{L} {+} \mathbf{S}_x B_p \sum_{j=1}^{n_\Delta} \tau_j \Delta^j  D_{xu}  ) \mathbf{s}_k  \nonumber \\
     & {+} \mathbf{S}_x (B_w {+} B_p \sum_{j=1}^{n_\Delta} \tau_j\Delta^j D_w )w_k  +\mathbf{S}_u  \sum_{j=1}^{n_\Delta} \tau_j \hat{\mathbf{u}}_{k+1}^j \label{eq:SuffCondPf1_1}\\
     &= \sum_{j=1}^{n_\Delta} \tau_j \bigg( \bigg( \mathbf{L}_K^j + \mathbf{S}_x B_p \Delta^j  D_{xu}  \nonumber\\ 
     &\quad + \mathbf{S}_u \mathbf{K}^j_{\Delta} \begin{bmatrix}
        I_{n_x+n_u}  & 0_{n_x+n_u, n_u(N-1)}
    \end{bmatrix}\bigg) \mathbf{s}_k \nonumber \\     
    & \quad + \textbf{S}_x \big( B_w + B_p\Delta^jD_w \big) w_k + \mathbf{S}_u\mathbf{M}^j B_w w_k
        \bigg) \label{eq:SuffCondPf1_2}\\
        & = \sum_{j=1}^{n_\Delta} \tau_j \bigg(\mathbf{C}_K^j \mathbf{s}_k + \mathbf{C}_M^j w_k  \bigg) . \label{eq:SuffCondPf1_3}
\end{align}
\end{subequations}
In \eqref{eq:SuffCondPf1}, \eqref{eq:SuffCondPf1_2} is obtained by substituting the control law \eqref{eq:InputParameterization} in \eqref{eq:SuffCondPf1_1} and using the definition of $\mathbf{L}_K^j$. Using \eqref{eq:SuffCondPf1}, it can be seen that
\begin{subequations}\label{eq:SuffCondPf2}
\begin{align}
   & H_{xu} \mathbf{S} \mathbf{s}_{k+1} = H_{xu}  \sum_{j=1}^{n_\Delta} \tau_j \bigg(\mathbf{C}_K^j \mathbf{s}_k + \mathbf{C}_M^j w_k  \bigg) \label{eq:SuffCondPf2_1}\\
   &=    \sum_{j=1}^{n_\Delta} \tau_j H_{xu} \big( \mathbf{C}_K^j \mathbf{s}_k + \mathbf{C}_M^j w_k  \big) \label{eq:SuffCondPf2_2}\\
   &\le  \sum_{j=1}^{n_\Delta} \tau_j (\mathbf{b}-\mathbf{t}) \le \mathbf{b}-\mathbf{t}, \label{eq:SuffCondPf2_3}
\end{align}
\end{subequations}
where \eqref{eq:SuffCondPf2_3} is obtained because any $\begin{bmatrix}
        \mathbf{s}_k \tr & w_k\tr
    \end{bmatrix}\tr $ inside $ \mathbb{P}_1$ also lies inside $\mathbb{P}_2^j$ from \eqref{eq:SuffCond}. 
Therefore, \eqref{eq:FeasCond_k+1} is satisfied for all feasible realizations of $\mathbf{s}_k$, $\Delta_k$ and $w_k$ under the proposed feedback if \eqref{eq:SuffCond} holds, and so the closed loop \eqref{eq:Dynamics}-\eqref{eq:OnlOptProb} satisfies SRF.
\end{proof}

Observe that unlike \eqref{eq:NecSuff}, \eqref{eq:SuffCond} does not depend on the state $x_k$. Thus, \eqref{eq:SuffCond} can be used to compute the MPC components offline. The following result based on Farkas' Lemma will be used to reformulate \eqref{eq:SuffCond}.
\vspace{0.3cm}
\begin{Lemma}\cite[Lemma~5.6]{kouvaritakis2015model}\label{Lem:Farkas}
    Let $\mathbb{X}_i = \Set{x}{H_i x \le h_i}$ for $i=1,2$ be two polytopes. Then, $\mathbb{X}_1 \subseteq \mathbb{X}_2 $ is satisfied if, and only if, there exists a matrix $\Lambda$ such that
    \begin{align}\label{eq:Farkas}
       \Lambda \ge 0, \quad \Lambda H_1 = H_2, \quad \Lambda h_1 \le h_2.
    \end{align}
\end{Lemma}
Note that the first inequality in \eqref{eq:Farkas} is element-wise.

\begin{Proposition}\label{Prop:LambdaRefo}
The set inclusions \eqref{eq:SuffCond} are satisfied if, and only if, there exist matrices $\{ \Lambda^j\}_{j=1}^{n_{\Delta}}$ such that 
    \begin{subequations}\label{eq:LambdaRefo}
    \begin{align}
        \Lambda^j \ge 0, \quad
    \Lambda^j &\begin{bmatrix}
        \mathbf{b}-\mathbf{t} \\ h_w
    \end{bmatrix} \le \mathbf{b}-\mathbf{t},
    \label{eq:LambdaRefo_1}\\
        \Lambda^j \begin{bmatrix}
        \mathbf{H}_{xu}\mathbf{S} & 0 \\
        0 & H_w
    \end{bmatrix} &= \mathbf{H}_{xu}\begin{bmatrix}
        \mathbf{C}_K^j &   \mathbf{C}_M^j
    \end{bmatrix},  \: \: \forall j \in \mathbb{N}_{1}^{n_{\Delta}}. \label{eq:LambdaRefo_2}
    \end{align}
    \end{subequations}
\end{Proposition}
\begin{proof}
    The proof follows from a direct application of Lemma \ref{Lem:Farkas} to each set inclusion in \eqref{eq:SuffCond}.
\end{proof}

Conditions \eqref{eq:LambdaRefo} will be used in the following section to build a robust MPC scheme satisfying SRF. 

\section{Design of robust MPC satisfying SRF}\label{Sec:RobustMPCDesign}
In this section, \eqref{eq:LambdaRefo} will be used as constraints in the design of the MPC components to ensure SRF and therefore constraint satisfaction. We present viable algorithmic solutions to enforce \eqref{eq:LambdaRefo} while minimizing the tightening of the constraints, so as to promote a larger region of attraction for the closed-loop problem.

\subsection{Terminal components}
It can already be seen that Theorem \ref{Th:RF_RCS} imposes the constraint $t_0\geq 0$. Additionally, Proposition \ref{Prop:LambdaRefo} explicitly formulates constraints under which the MPC problem satisfies SRF. The MPC components in \eqref{eq:OnlOptProb} which are to be computed are  $Y, z, \mathbf{t}$, and the terminal cost $l_N(\cdot)$. The design for SRF does not depend on the choice of the cost function, and hence, $l_N(\cdot)$ will be defined later. 

Ideally, all the components $Y, z, \mathbf{t}$ would be set as decision variables in the offline design problem. However, this would result in large number of bilinear constraints in \eqref{eq:LambdaRefo}, as  $Y$ is a component of $\mathbf{H}_{xu}$ defined in \eqref{eq:Notation1}. Thus, in this work, a popular strategy from the literature is used to choose $Y$ and $z$ a priori \cite{kouvaritakis2015model}. Specifically, they are chosen to be of the structure of maximal robust positively invariant sets under a terminal feedback law $K_Y$, given as
\begin{align}\label{eq:TermSetPara}
    Y := \begin{bmatrix}
        F+G K_Y \\
        (F+G K_Y) (A+B K_Y)\\
        \vdots
        \\ 
        (F+G K_Y) (A+B K_Y)^{k'}
    \end{bmatrix}, z:=\begin{bmatrix}
        b \\
        b \\
        \vdots \\
        b
    \end{bmatrix},
\end{align}
where $K_Y$ is a feedback chosen such that $A+BK_Y$ is Schur stable, and $k'$ is a positive scalar chosen by the designer depending on the desired flexibility.  Although the components $Y$ and $z$ are fixed, the tightenings $t_N$ allow the optimizer to modify the terminal set.

\subsection{Offline optimization for constraint tightening}
Given an uncertain model \eqref{eq:Dynamics}-\eqref{eq:wBound} and constraint set \eqref{eq:Constraints}, the terminal set is computed as in \eqref{eq:TermSetPara}, and the following optimization problem can be solved to compute the tightenings, $\mathbf{t}$. 
\begin{align}\label{eq:OfflineOpt_gen}
\begin{split}
    \min_{\substack{\mathbf{t}, \\
   \{ K^j, \Lambda^j,  \mathbf{M}^j,    \mathbf{K}^j_{\Delta} \}_{j=1}^{n_{\Delta}} }} &l_t(\mathbf{t})\\
    \text{s.t.} \qquad \eqref{eq:LambdaRefo}, \quad &t_0 \ge 0, \quad f_t(\mathbf{t}) \le 0.
\end{split}
\end{align}
In \eqref{eq:OfflineOpt_gen}, $l_t(\cdot)$ represents an objective function, and $f_t(\cdot)$ represents additional constraints which can be imposed on the tightenings based on the desired goal of the control application. Moreover, the constraint $t_0 \ge 0$ needs to be explicitly enforced, so that recursive feasibility of \eqref{eq:OnlOptProb} results in robust constraint satisfaction, as shown in Theorem \ref{Th:RF_RCS}. 
Note that the guarantees derived in Theorem \ref{Thm:SuffCond} and Proposition \ref{Prop:LambdaRefo} are independent of $l_t(\cdot)$ and $f_t(\cdot)$, as long as the constraints \eqref{eq:LambdaRefo} are satisfied.

Here, we describe one possible strategy of choosing $l_t(\cdot)$ and $f_t(\cdot)$, where the design objective is to maximize the region of attraction (ROA) of the closed-loop system. The ROA is defined by the feasible region of the state for the MPC optimization \eqref{eq:OnlOptProb}, and thus is a convex polytope. Because the volume of a generic polytope is difficult to compute  even when the hyperplanes are known a priori \cite{lawrence1991polytope}, maximizing ROA is a difficult task. Under the chosen structure of the optimization problem \eqref{eq:OnlOptProb}, minimization of $\norm{\mathbf{t}}_2^2$ can be used as an approximation for maximizing ROA as proposed in \cite{parsi2022computationally}. 

However, solely approximating ROA maximization with minimization of $\norm{\mathbf{t}}_2^2$ can sometimes lead to poor results. In this work, the approximation is improved by additionally maximizing the size of an $l_1$-norm ball that can fit inside the feasible region of the state space. Consider the vertices of the unit-$l_1$-norm ball in $\mathbb{R}^{n_x}$, represented by $\{x^i\}_{i=1}^{2n_x}$, which lie on the positive and negative directions of the principal axes of the state space. For each vertex of the $l_1$-norm ball, define the corresponding MPC optimization variable $\mathbf{s}^i_{\alpha}$ as
\begin{equation}
    \mathbf{s}^{i}_{\alpha} = \begin{bmatrix}
        \alpha {x^i} \\ \mathbf{u}^i
    \end{bmatrix}, \quad \forall i \in \mathbb{N}_{1}^{2n_x}
\end{equation}
where $\alpha$ is a positive scalar defining the size of the $l_1$-norm ball, and $\mathbf{u}^i$ represents a feasible input sequence for the $i^{\text{th}}$ vertex. Then, the offline optimization problem can be written as
\begin{subequations}\label{eq:OfflineOpt}
\begin{align}
    \min_{\substack{\mathbf{t},  \alpha, \{\mathbf{u}^{i}\}_{i=1}^{2n_x}, \\
   \{ \Lambda^j, K^j,   \mathbf{M}^j,    \mathbf{K}^j_{\Delta} \}_{j=1}^{n_{\Delta}} }} &\norm{\mathbf{t}}^2_2 - \mu \alpha \label{eq:OfflineOpt1}\\
    \text{s.t.} \qquad &     \Lambda^j \ge 0, \quad \alpha > 0, \quad t_0 \ge 0, \label{eq:OfflineOpt2}\\
        \Lambda^j \begin{bmatrix}
        \mathbf{H}_{xu}\mathbf{S} & 0 \\
        0 & H_w
    \end{bmatrix} &= \mathbf{H}_{xu}\begin{bmatrix}
        \mathbf{C}_K^j &   \mathbf{C}_M^j
    \end{bmatrix}, \label{eq:OfflineOpt3}\\
    \Lambda^j \begin{bmatrix}
        \mathbf{b}-\mathbf{t} \\ h_w
    \end{bmatrix} &\le \mathbf{b}-\mathbf{t}, \quad \forall j \in \mathbb{N}_{1}^{n_{\Delta}}, \label{eq:OfflineOpt4}\\
    \mathbf{H}_{xu} \mathbf{S} \mathbf{s}^{i}_{\alpha} &\le \mathbf{b-t}, \quad \forall i \in \mathbb{N}_{1}^{2n_x}, \label{eq:OfflineOpt5}
\end{align}
\end{subequations}
where  $\mu$ is a positive weighting factor between minimizing $\norm{\mathbf{t}}_2$ and maximizing $\alpha$.
Additionally, \eqref{eq:OfflineOpt5} represent the feasibility constraints for each vertex of the $l_1$-norm ball. 

The optimization problem \eqref{eq:OfflineOpt} is non-convex, due to the presence of bilinear constraints \eqref{eq:OfflineOpt4}.  Although this increases the complexity of design,  problem \eqref{eq:OfflineOpt} needs to be solved only once, in the offline phase. Note that the number of bilinearities is greatly reduced because $Y$ and $z$ are chosen beforehand, and consequently \eqref{eq:OfflineOpt3} are affine constraints. Moreover, given a good initial guess of the tightenings $\mathbf{t}$, feasible solutions to all the other optimization variables can be computed by solving a convex optimization problem (by fixing $\mathbf{t}$ in \eqref{eq:OfflineOpt}). This provides a simple way to initialize non-convex optimization techniques which can solve \eqref{eq:OfflineOpt}. 

\vspace{0.1cm}
\begin{Remark}[Initial guess]\label{Rem:InGuess}
It can be difficult to compute a feasible initial guess for the tightenings  $\mathbf{t}$ in \eqref{eq:OfflineOpt}, as the tightenings are problem dependent. In practice, one strategy which produced good results was to use an existing constraint tightening technique, such as \cite{chisci2001systems} or \cite{parsi2022computationally} to guess $\mathbf{t}$. Note that these techniques were designed for systems affected by additive disturbances alone, and may not always result in feasible initial guesses for $\mathbf{t}$. 
\end{Remark}

\subsection{Cost functions and online optimization}
Whereas the constraints in \eqref{eq:OnlOptProb} are designed to ensure recursive feasibility, the cost function must be chosen to guarantee closed-loop stability. A quadratic tracking cost formulation is presented here for simplicity, but the method proposed here can easily be extended to economic cost functions.
Consider the stage cost defined as
\begin{equation}
    l(x,u) = x\tr Q_x x + u\tr Q_u u,
\end{equation}
where the matrices $Q_x, Q_u$ are positive definite. The design task is to choose a terminal cost of the form $l_N(x) = x\tr Q_N x$ such that the closed-loop system is stable. This can be achieved by choosing $Q_N$ to be a positive definite matrix satisfying $ \forall j\in \mathbb{N}_{1}^{n_\Delta}$
\begin{equation}\label{eq:TermCostLyap}
     (1+\epsilon)\mathbf{C}_K^{j\intercal}  Q_s \mathbf{C}_K^j  - \mathbf{S}\tr  Q_s \mathbf{S} \prec -\begin{bmatrix}
         P_x & 0 \\ 0 & P_u
     \end{bmatrix}  = -P,
\end{equation}
where $Q_s:= \text{diag}\{Q_x \otimes I_N, Q_N, Q_u\otimes I_N \}$, $\epsilon\in \mathbb{R}_{>0}$ is a design parameter and $P_x\in \mathbb{R}^{n_x \times n_x}$, $P_u\in \mathbb{R}^{Nn_u \times Nn_u}$ are positive definite matrices. The values of $\epsilon$ and $P$ affect the resulting stability margin as will be shown Section \ref{Sec:AlgProp}. For a given $\epsilon$ and known values of the feedback gains $\{ K^j, \mathbf{M}^j,    \mathbf{K}^j_{\Delta} \}_{j=1}^{n_{\Delta}}$, constraints \eqref{eq:TermCostLyap} are linear matrix inequalities in the variables $Q_N$ and $ P$. There are multiple strategies which can be used to choose $Q_N$ such that \eqref{eq:TermCostLyap} holds, for example, by minimizing the trace of $Q_N$ or minimizing the deviation of $Q_N$ from the infinite horizon LQR gain.

The online optimization can then be written as
\begin{subequations}\label{eq:OnlOptProb_reg}
    \begin{align}
        \min_{\hat{\mathbf{x}}_k,\hat{\mathbf{u}}_k}  \quad  \hat{x}_{N|k}\tr Q_N \hat{x}_{N|k} + &\displaystyle \sum_{i=0}^{N-1} \left( \hat{x}_{i|k}\tr Q_x \hat{x}_{i|k} +  \hat{u}_{i|k}\tr Q_u \hat{u}_{i|k} \right) \label{eq:OnlOptProb_reg1} \\
		\text{s.t.} \quad A\hat{x}_{i|k} + B\hat{u}_{i|k} &= \hat{x}_{i|k+1}, \quad \hat{x}_{k,0}=x_k, \label{eq:OnlOptProb_reg2}\\
            F\hat{x}_{i|k} + G\hat{u}_{i|k}  &\le b - t_i,  \quad i\in \mathbb{N}_{0}^{N-1}, \label{eq:OnlOptProb_reg3} \\
            Y\hat{x}_{N|k} &\le z - t_N. \label{eq:OnlOptProb_reg4}
    \end{align}
\end{subequations}
The computational complexity of \eqref{eq:OnlOptProb_reg} is the same as that of a nominal MPC problem. This enables a large computational advantage compared to existing robust MPC methods, such as tube MPC \cite{schwenkel2022}. 

\subsection{Algorithm and properties}\label{Sec:AlgProp}
Using the offline \eqref{eq:OfflineOpt} and online \eqref{eq:OnlOptProb_reg} optimization problems, a robust MPC algorithm is formally described in Algorithm \ref{Alg:RMPC}. It can be seen that the offline phase involves the computation of $\mathbf{t}$ by solving  \eqref{eq:OfflineOpt}, and subsequently choosing the terminal cost matrix such that \eqref{eq:TermCostLyap} holds.  Although a less conservative approach would be to include \eqref{eq:TermCostLyap} as an additional constraint in \eqref{eq:OfflineOpt}, this results in a semi-definite program with bilinear constraints, and is difficult to solve. An alternative approach is to first compute $Q_N$ and $\{ K^j, \mathbf{M}^j,    \mathbf{K}^j_{\Delta} \}_{j=1}^{n_{\Delta}}$ satisfying \eqref{eq:TermCostLyap}, and then solving \eqref{eq:OfflineOpt} with the resulting control gains. 

\begin{algorithm}[h]
\caption{Closed-loop robust MPC}\label{Alg:RMPC}
\begin{algorithmic}[1]
\Statex \textbf{Offline:} 
\State Choose $N, K_Y$, and compute $Y$ and $z$
\State Choose $\mu$ and $\epsilon$
\State Solve  \eqref{eq:OfflineOpt}  to compute $\mathbf{t}$ and control gains
\State Choose $Q_N$ satisfying \eqref{eq:TermCostLyap} 
\setcounter{ALG@line}{0}
\Statex \textbf{Online:} At each time-step $k\ge0$:
\State Obtain the measurement $x_k$
\State Solve \eqref{eq:OnlOptProb_reg}
\State Apply $\pi(x_k)=\hat{u}^*_{k,0}$ 
\end{algorithmic}
\end{algorithm}

The properties of the proposed algorithm will now be discussed. For this purpose, the notion of regional input-to-state stability (ISS) is first defined.
\begin{Definition}[Regional ISS in $\mathbb{X}$ \cite{limon2009input}]
    Given a system with dynamics \eqref{eq:Dynamics}-\eqref{eq:wBound} and a compact set
$\mathbb{X}\subseteq \mathbb{R}^{n_x}$ including the origin as an interior point, the system is said to be regionally ISS in $\mathbb{X}$ with
respect to $w_k$ if $\mathbb{X}$ is a robust positively invariant set and if there exist a $\mathcal{KL}$ function $\beta(\cdot, \cdot)$ and a $\mathcal{K}$ function $\lambda_1(\cdot)$  such that, for all $x_0 \in \mathbb{X}$ and $k \ge 0$
\begin{equation}\label{eq:ISS}
    \norm{x_k} \le \beta(\norm{x_0},k) + \lambda_1(\norm{\mathbf{w}_{[0:k-1]}}) ,
\end{equation}
where $\norm{\mathbf{w}_{[0:k-1]}} = \sup_{i\in\mathbb{N}_{0}^{k-1}} \norm{w_i}$.
\end{Definition}
A sufficient condition to guarantee ISS of a system is the existence of an ISS Lyapunov function \cite{limon2009input}, defined below. 
\begin{Definition}
    Consider the system \eqref{eq:Dynamics}-\eqref{eq:wBound} and a compact set $\mathbb{X}\subseteq \mathbb{R}^{n_x}$ including the origin as an interior point. A function $V:\mathbb{R}^n \rightarrow \mathbb{R}_{\ge 0}$ is called an ISS Lyapunov function in $\mathbb{X}$ with respect to $w_k$ if there exist $\mathcal{K}_{\infty}$ functions $\sigma_1 (\cdot),\sigma_2(\cdot),\sigma_3(\cdot)$ and a $\mathcal{K}$ function $\lambda_2(\cdot)$ such that $\forall x_k \in \mathbb{X}$
    \begin{subequations}\label{eq:ISSLyap}
    \begin{align}
    \sigma_1(\norm{x_k}) \le V(x_k) &\le \sigma_2(\norm{x_k}), \label{eq:ISSLyap1} \\
    V(x_{k+1}) - V(x_k)  &\le -\sigma_3(\norm{x_k}) + \lambda_2(w_k).\label{eq:ISSLyap2}
    \end{align}
    \end{subequations}
\end{Definition}

\begin{Theorem}\label{Thm:Props}
    Let the offline optimization problem \eqref{eq:OfflineOpt} be feasible, the cost function satisfy \eqref{eq:TermCostLyap}, and the online optimization \eqref{eq:OnlOptProb_reg} have a feasible solution at time $k=0$. Moreover, let $\mathbb{P}_{0,x}$ be the projection of $\mathbb{P}_0$ on the first $n_x$ elements of $\mathbf{s}_k$, and let the origin be an interior point of $\mathbb{P}_{0,x}$. Then, the closed loop formed by a system with dynamics \eqref{eq:Dynamics}-\eqref{eq:wBound} and the robust MPC controller in Algorithm \ref{Alg:RMPC} satisfies the following properties:
    \begin{enumerate}[label=(\alph*)]
        \item The problem \eqref{eq:OnlOptProb_reg} remains feasible for all $k > 0$. \label{Thm:Props1}
        \item The constraints \eqref{eq:Constraints} are satisfied for all $k > 0$.\label{Thm:Props2}
        \item The closed-loop system is ISS in $\mathbb{P}_{0,x}$ with respect to $w_k$.\label{Thm:Props3}
    \end{enumerate} 
\end{Theorem}
\begin{proof}
    In order to prove property \ref{Thm:Props1}, note that the constraints in Proposition \ref{Prop:LambdaRefo} are satisfied because the offline optimization \eqref{eq:OfflineOpt} is feasible. Therefore, the set inclusions  \eqref{eq:SuffCond} are satisfied, and Theorem \ref{Thm:SuffCond} implies that the online MPC optimization \eqref{eq:OnlOptProb_reg} is strongly recursively feasible.  Moreover, because \eqref{eq:OnlOptProb_reg} admits a feasible solution at time $k=0$, it remains feasible for all $k > 0$. 
    
    The property \ref{Thm:Props2} follows from Theorem \ref{Th:RF_RCS}. This is because the offline optimization \eqref{eq:OfflineOpt} explicitly enforces $t_0 \ge 0$, and SRF implies RF. 

    The proof of \ref{Thm:Props3} can be seen as follows. Let the optimal MPC cost function be a candidate ISS Lyapunov function denoted by
    \begin{align}\label{eq:CompactCost}
    \begin{split}
        V(x_k) &= \hat{x}_{N|k}^{*\intercal} Q_N \hat{x}_{N|k}^* + \displaystyle \sum_{i=0}^{N-1} \left( \hat{x}_{i|k}^{*\intercal} Q_x \hat{x}_{i|k}^* +  \hat{u}_{i|k}^{*\intercal} Q_u \hat{u}_{i|k}^* \right) \\
        &= \mathbf{s}_k^{*\intercal} \mathbf{S}\tr  Q_s \mathbf{S} \mathbf{s}_k^*.
    \end{split}
    \end{align}
    
    In the following, it will be shown that $V(x)$ satisfies \eqref{eq:ISSLyap}.  Consider  $\sigma_1(\norm{x}) = \rho_{\min} (Q_x) \norm{x}^2$, which is guaranteed to be a lower bound of $V(x)$. Moreover, $V(x)$ is a continuous and piecewise quadratic function in  $\mathbb{P}_{0,x}$ as shown in \cite{bemporad2002explicit}. Therefore, the existence of $\sigma_2(\cdot)$ satisfying \eqref{eq:ISSLyap1} is a direct consequence of \cite[Lemma~4]{limon2009input}.
    
    In order to show \eqref{eq:ISSLyap2}, let the perturbation $\Delta_k$ be given as \eqref{eq:Delta_k} and consider the candidate input sequences in \eqref{eq:InputParameterization}. The cost function satisfies
    \begin{subequations}\label{eq:Lyapunov}
    \begin{align}
        &V(x_{k+1}) - V(x_k) = \mathbf{s}_{k+1}^{*\intercal} \mathbf{S}\tr  Q_s \mathbf{S} \mathbf{s}_{k+1}^* - \mathbf{s}_k^{*\intercal} \mathbf{S}\tr  Q_s \mathbf{S} \mathbf{s}_k^* \nonumber \\
        &\le \sum_{j=1}^{n_\Delta} \tau_j \bigg(\mathbf{C}_K^j \mathbf{s}_k^* + \mathbf{C}_M^j w_k  \bigg)\tr  Q_s \sum_{l=1}^{n_\Delta} \tau_l \bigg(\mathbf{C}_K^l \mathbf{s}_k^* + \mathbf{C}_M^l w_k  \bigg)\nonumber \\
        & \qquad \qquad - \mathbf{s}_k^{*\intercal} \mathbf{S}\tr  Q_s \mathbf{S} \mathbf{s}_k^* \label{eq:Lyapunov1}\\
        &\le  \left(1+\frac{1}{\epsilon}\right) \left(w_k\tr \mathbf{C}_M(\tau)\tr Q_s \mathbf{C}_M(\tau) w_k \right) + \nonumber\\ 
        & \quad (1+\epsilon) \left(\mathbf{s}_k^{*\intercal} \mathbf{C}_K(\tau)\tr  Q_s \mathbf{C}_K(\tau) \mathbf{s}_k^* \right) - \mathbf{s}_k^{*\intercal} \mathbf{S}\tr  Q_s \mathbf{S} \mathbf{s}_k^* \label{eq:Lyapunov2} \\
        &\le \lambda_2(w_k) + \mathbf{s}_k^{*\intercal} ( (1+\epsilon) \mathbf{C}_K(\tau)\tr  Q_s \mathbf{C}_K(\tau)  - \mathbf{S}\tr  Q_s \mathbf{S} )\mathbf{s}_k^* \label{eq:Lyapunov3}
    \end{align}
    \end{subequations}
     where $\mathbf{C}_K(\tau) := \sum_{j=1}^{n_\Delta} \tau_j\mathbf{C}_K^j $ and $\mathbf{C}_M(\tau) := \sum_{j=1}^{n_\Delta} \tau_j\mathbf{C}_M^j $, and \eqref{eq:Lyapunov2} is obtained by applying Cauchy-Schwarz theorem and Fenchel's inequality \cite{boyd2004convex}, by which, $\norm{x+w} \le (1+\epsilon)\norm{x} + (1+\frac{1}{\epsilon}) \norm{w}$ for any $\epsilon>0$. The function $\lambda_2(w_k)$ in \eqref{eq:Lyapunov3} can chosen as the first term in \eqref{eq:Lyapunov2}. Moreover, using \eqref{eq:TermCostLyap}, it can be seen that 
     \begin{subequations}\label{eq:SchurTerm}
     \begin{align}
         &\mathbf{S}\tr  Q_s \mathbf{S} - P - (1{+}\epsilon)\mathbf{C}_K^{j\intercal}  Q_s \mathbf{C}_K^j  \succ 0, \quad \forall j \in \mathbb{N}_{1}^{n_\Delta} \label{eq:SchurTerm1}\\
         &\implies \begin{bmatrix}
             (1{+}\epsilon)^{-1}Q_s^{-1}  & \mathbf{C}_K^j \\
            \mathbf{C}_K^{j\intercal} & \mathbf{S}\tr  Q_s \mathbf{S} - P
         \end{bmatrix} \succ 0, \: \forall j \in \mathbb{N}_{1}^{n_\Delta}, \label{eq:SchurTerm2}\\
        &\implies  \begin{bmatrix}
            \displaystyle\sum_{j=1}^{n_\Delta} \tau_j (1{+}\epsilon)^{-1} Q_s^{-1}  & \displaystyle\sum_{j=1}^{n_\Delta} \tau_j \mathbf{C}_K^j \\
            \displaystyle\sum_{j=1}^{n_\Delta} \tau_j \mathbf{C}_K^{j\intercal} & \displaystyle\sum_{j=1}^{n_\Delta} \tau_j (\mathbf{S}\tr  Q_s \mathbf{S} - P)
         \end{bmatrix} \succ 0, \label{eq:SchurTerm3}\\
       &\implies  \begin{bmatrix}
             (1{+}\epsilon)^{-1}Q_s^{-1}  & \mathbf{C}_K(\tau) \\
            \mathbf{C}_K(\tau)^\intercal & (\mathbf{S}\tr  Q_s \mathbf{S}- P)
         \end{bmatrix}  \succ 0, \label{eq:SchurTerm4}\\
        &\implies  \mathbf{S}\tr  Q_s \mathbf{S}- P - (1{+}\epsilon)\mathbf{C}_K(\tau)\tr  Q_s \mathbf{C}_K(\tau)   \succ 0, \label{eq:SchurTerm5}
     \end{align}
     \end{subequations}
     where the Schur complement lemma \cite{zhang2006schur} is used to obtain \eqref{eq:SchurTerm2} and \eqref{eq:SchurTerm5}. Using \eqref{eq:Lyapunov3} and \eqref{eq:SchurTerm5},
\begin{align}
\begin{split}
    &V(x_{k+1}) - V(x_k) \le \lambda_2(w_k) - s_k^{*\intercal} P s_k^* \\
    &\le \lambda_2(w_k) - x_k^{\intercal}P_x x_k = \lambda_2(w_k) - \sigma_3(\norm{x_k}). 
\end{split}
\end{align}
Thus, $V(x)$ satisfies \eqref{eq:ISSLyap}, and is an ISS Lyapunov function of the system. Therefore, applying \cite[Lemma~3.5]{jiang2001input}, the closed-loop system is ISS.
\end{proof}
It can be seen that $\epsilon$ is used to define $\lambda_2(\cdot)$, which affects the function $\lambda_1(\cdot)$ in \eqref{eq:ISS} \cite{jiang2001input}. Thus, the choice of $\epsilon$ affects the stability margin of the robust MPC controller, and can be used as a tuning parameter by the designer. 

\subsection{Discussion}
The central idea of this work is that RF is a one-step property, and can be used to guarantee constraint satisfaction in closed loop as shown in Theorem \ref{Th:RF_RCS}. However, barring a few results in early MPC literature such as \cite{chisci2001systems,kerrigan2001Robust}, this idea has not been fully utilized. Hence, some discussion is warranted on interesting research directions.

\subsubsection{Beyond constraint tightening}
Theorem \ref{Th:RF_RCS} does not restrict the robust MPC optimization to be of the structure \eqref{eq:OnlOptProb_reg}, which is a characteristic of constraint tightening techniques. The theorem is also valid for any MPC optimization of the form
\begin{subequations}\label{eq:OnlOptProb_flex}
    \begin{align}
        \min_{\hat{\mathbf{x}}_k,\hat{\mathbf{u}}_k}  \quad  \hat{x}_{N|k}\tr &Q_N \hat{x}_{N|k} + \textstyle\sum_{i=0}^{N-1}  \hat{x}_{i|k}\tr Q_x \hat{x}_{i|k} +  \hat{u}_{i|k}\tr Q_u \hat{u}_{i|k} \label{eq:OnlOptProb_flex1} \\
		\text{s.t.} \quad &\hat{x}_{k,0}=x_k,\\
            A&\hat{x}_{i|k} + B\hat{u}_{i|k} = \hat{x}_{i|k+1}, \:  \forall i\in \mathbb{N}_{0}^{N-1},\label{eq:OnlOptProb_flex2}\\
            &F\hat{x}_{0|k} + G\hat{u}_{0|k}  \le b,   \label{eq:OnlOptProb_flex3} \\
            &\Phi_x \hat{x}_{0|k} + \Phi_u \hat{\mathbf{u}}_k \le \phi, \label{eq:OnlOptProb_flex4}
    \end{align}
\end{subequations}
where $\Phi_x,\Phi_u$ and $\phi$ can be chosen by the designer. In \eqref{eq:OnlOptProb_flex}, the system constraints \eqref{eq:Constraints} are only imposed on the first state and input in the prediction horizon in \eqref{eq:OnlOptProb_flex3}, which can be seen as selecting $t_0 = 0$ in \eqref{eq:OnlOptProb}. Using a similar approach as in Theorem \ref{Thm:SuffCond} and Proposition \ref{Prop:LambdaRefo}, sufficient conditions can also be derived for SRF of \eqref{eq:OnlOptProb_flex}. However, in practice, solving the resulting non-convex optimization problem is harder, due to the higher number of bilinearities. 

\subsubsection{Robust control invariance}
A key concept which is known in control literature is that of robust control invariance (RCI) \cite{blanchini2008set,kerrigan2001robust_PHD}. As opposed to  robust invariance (RI), which considers invariance of autonomous systems subject to uncertainty, RCI considers controlled dynamical systems. For this reason, computation of RCI sets is challenging \cite{rakovic2007optimized,rungger2017}.  

RCI provides an alternative interpretation of Theorem \ref{Thm:NecSuff}. A necessary and sufficient condition for strong recursive feasibility of MPC is that the feasible region is RCI. Because generic RCI sets are difficult to compute, the proposed method uses parameterized feedback laws to formulate the problem as computation of the resulting parameterized RI sets. Note that all RI sets are also RCI sets, and are easier to compute in practice. 

\section{Illustrative example}
In this section, Algorithm \ref{Alg:RMPC} will be used to generate a controller for a mass-spring-damper system, where two masses are connected to each other by a spring and a damper, and the spring constant and damping coefficient are uncertain.

The example has been adapted from \cite{parsi2022scalable}, where similar systems with increasing number of masses are considered, and a tube MPC strategy was used to design the controller. The state of the system is a vector of the  position and velocity of each mass, and the control input is a vector of forces acting on each mass.  The system dynamics can be represented in discrete time by 
\begin{align}\label{eq:MSD_Dyn}
A &= \left[ \arraycolsep=4pt \begin{array}{cccc}
1 & T_s & 0  & 0 \\
{-}\frac{k_{12}T_s}{m_1} & {-}\frac{c_{12}T_s}{m_1}{+}1 & \frac{k_{12}T_s}{m_1} & \frac{c_{12}T_s}{m_1} \\ 
0 & 0 & 1  & T_s\\
\frac{k_{12}T_s}{m_2} & \frac{c_{12}T_s}{m_2} & {-}\frac{k_{12}T_s}{m_2} & -\frac{c_{12}T_s}{m_2} {+} 1 \\
\end{array}\right], \nonumber\\
B_p &{=} \left[ \begin{array}{cccc}
0 & 0   \\
\frac{k_{u}T_s}{m_1} & \frac{c_u T_s}{m_1}  \\
0 & 0 \\
{-}\frac{k_{u}T_s}{m_2} & {-}\frac{c_{u}T_s}{m_2}  \\
\end{array} \right] , 
B {=} 
\begin{bmatrix}
0 & 0  \\
\frac{T_s}{m_1} & 0  \\
0 & 0  \\
0 & \frac{T_s}{m_2} 
\end{bmatrix}, 
B_w {=} w_b B, 
\nonumber \\%
D_x &= 
\begin{bmatrix}
{-}1 & 0 \\
0 & 1  \\
\end{bmatrix}, \quad D_u = 0, \quad D_w = 0,
\end{align}
where $m_1 = m_2 =0.2\mathrm{kg}$ is the mass, $k_{12} = 0.5 \mathrm{Nm^{-1}}$ and $c_{12} = 0.5 \mathrm{Nsm^{-1}}$ are the nominal spring constant and damping coefficient. 
In addition, $T_s = 0.1\mathrm{s}$ represents the sampling time used for discretization, and $w_b=0.2$ represents the bound on a disturbance acting on the velocity states.  The magnitudes of all the states and inputs are bounded by 2. The terms $k_u {=} 0.04 k_{12}$ and $c_u = 0.02 c_{12}$  model the structured uncertainty matrix $B_p$. The matrices $\Delta^j$ defining $\mathcal{P}$ in \eqref{eq:DeltaBound} were chosen as diagonal matrices of size $ 2\times 2$ with $\pm1$ on each diagonal entry, modeling a $\pm4\%$ uncertainty for the spring constant, and $\pm2\%$ uncertainty for the damping coefficient.

Two controllers are designed for the system. The first uses Algorithm \ref{Alg:RMPC}, and is referred to as CLR, for closed-loop robust MPC. The second uses the method proposed in \cite{parsi2022scalable}, where an ellipsoidal tube is constructed at each time step by the controller to bound the states and inputs in the prediction horizon. This method will be referred to as ET, for ellipsoidal tube MPC. Both the controllers use $N=5$ time steps. The terminal constraints of CLR are computed using $K_Y$ chosen as the LQR gain, and $k' = 2$. The constants $\mu$ and $\epsilon$ were chosen as 2 and 0.1 respectively. The terminal cost is computed such that it satisfies \eqref{eq:TermCostLyap} and its deviation from the infinite horizon LQR gain is minimized.

The initial state of the system is chosen to be $x_0 = [1.9,0.5,-1.7,1.7]\tr$. The disturbance affecting the system is the same for both the controllers, and is sampled from a uniform distribution in $\mathcal{W}$. The simulations were performed for 25 realizations of the spring constants, damping coefficients and disturbance sequences chosen using a uniform distribution within the specified bounds. 

It was observed that CLR requires higher offline computation time compared to ET, and a much lower online computation time. This is because the ET algorithm solves semi definite programs offline and online, whereas CLR  solves a non-convex optimization offline and a quadratic program online. All the optimization problems were solved  on a laptop with Intel i7-8550U processor, and formulated using YALMIP \cite{Lofberg2004}. The convex optimization problems were solved using Gurobi \cite{gurobi} and MOSEK \cite{mosek}, and the non-convex optimization using IPOPT \cite{wachter2006implementation}. The offline computation time for the CLR controller is  $ 45.1 \mathrm{s}$, whereas that for ET is $ 0.06 \mathrm{s}$. For solving \eqref{eq:OfflineOpt}, an initial guess for the tightenings was computed as follows. Ignoring the perturbation from $\Delta_k$, the approach from \cite{parsi2022computationally} was used to obtain a tightening. This tightening was then scaled by a factor of 1.7, which was a feasible initial guess for \eqref{eq:OfflineOpt}. The average online computation time for CLR is $ 2.2 \mathrm{ms}$, whereas that for ET is $ 53.9 \mathrm{ms}$. 

The closed-loop trajectories of the system are shown in Figure \ref{fig:traj_msd}, where highlighted regions show the upper and lower bounds of state and input variables at each time step in closed loop over all the simulations. It can be seen that both controllers successfully regulate the state to the origin without constraint violations. However, the ET controller is aggressive at the start of the simulation, because the initial position is too close to the constraint, and the corresponding velocity is positive. Moreover, the ET controller recursively outer-approximates the effect of worst-case perturbations on state trajectories over the prediction horizon, resulting in the aggressive behavior. The average closed-loop cost achieved by CLR is 83.0 and that for ET is 93.7, which is a $11.4\%$ cost reduction.   

\begin{figure}[t]
	\centering
	\includegraphics[scale=1]{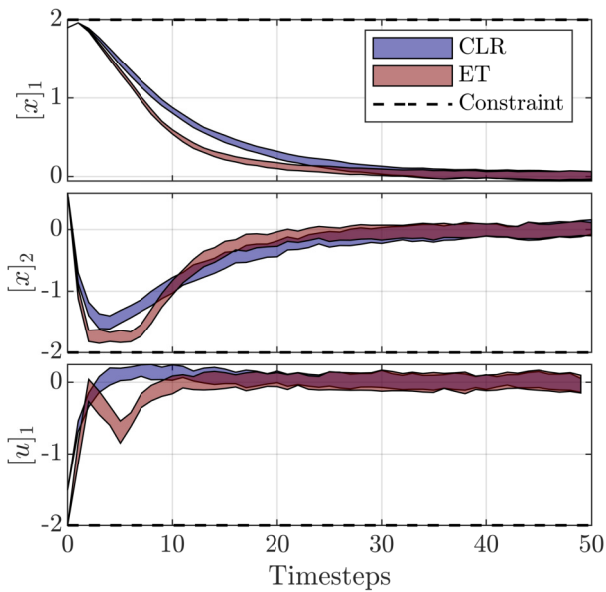}
	\caption{ Upper and lower bounds on the closed-loop trajectories with CLR and ET MPC algorithms over 25 random realizations of disturbances and true system parameters.} 
 \label{fig:traj_msd}
\end{figure}

\section{Conclusions}
A novel design approach was presented for robust MPC, using the property that recursive feasibility can imply constraint satisfaction in closed loop. Recursive feasibility is a one-step property, meaning that it depends on perturbations occurring at \emph{one} time step. Moreover, recursive feasibility only requires the existence of a feasible input sequence at the subsequent time step, as the MPC optimizer can compute the input sequence in closed loop. Exploiting these ideas, it was shown that recursive feasibility can be imposed as a constraint in the robust MPC design. Such a design reduces the conservatism of most existing MPC methods, which ensure constraint satisfaction by computing feasible open-loop trajectories under multiple perturbations and disturbances occurring along the prediction horizon. 

\bibliographystyle{IEEEtran}
\bibliography{biblio_CLR}   

\end{document}